\documentclass[11pt]{article}
\usepackage[utf8x]{inputenc}
\usepackage[numbers,sort,compress]{natbib}
\usepackage[french,english]{babel}
\usepackage{graphicx}
\usepackage{fullpage}
\usepackage{chngcntr}
\usepackage[colorlinks=true,linkcolor=blue,citecolor=magenta,linktocpage=true]{hyperref}
\usepackage{amsfonts,amssymb,amsthm,bbm,mathtools}
\usepackage{braket}
\usepackage[usenames,dvipsnames]{xcolor}
\usepackage{url}
\usepackage{afterpage}
\usepackage[outdir=./]{epstopdf}
\usepackage{caption}
\usepackage{wrapfig}
\usepackage[bottom]{footmisc}
\usepackage{cancel}
\usepackage{tocloft}
\usepackage{pict2e}
\makeatletter
\DeclareRobustCommand{\loplus}{\mathbin{\mathpalette\dog@lsemi{+}}}
\DeclareRobustCommand{\lotimes}{\mathbin{\mathpalette\dog@lsemi{\times}}}
\DeclareRobustCommand{\roplus}{\mathbin{\mathpalette\dog@rsemi{+}}}
\DeclareRobustCommand{\rotimes}{\mathbin{\mathpalette\dog@rsemi{\times}}}

\newcommand{\dog@rsemi}[2]{\dog@semi{#1}{#2}{-90,90}}
\newcommand{\dog@lsemi}[2]{\dog@semi{#1}{#2}{270,90}}
\newcommand{\dog@semi}[3]{%
  \begingroup
  \sbox\z@{$\m@th#1#2$}%
  \setlength{\unitlength}{\dimexpr\ht\z@+\dp\z@\relax}%
  \makebox[\wd\z@]{\raisebox{-\dp\z@}{%
    \begin{picture}(1,1)
    \linethickness{\variable@rule{#1}}
    \roundcap
    \put(0.5,0.5){\makebox(0,0){\raisebox{\dp\z@}{$\m@th#1#2$}}}
    \put(0.5,0.5){\arc[#3]{0.5}}
    \end{picture}%
  }}%
  \endgroup
}
\newcommand{\variable@rule}[1]{%
  \fontdimen8  
  \ifx#1\displaystyle\textfont3\else
    \ifx#1\textstyle\textfont3\else
      \ifx#1\scriptstyle\scriptfont3\else
        \scriptscriptfont3\relax
  \fi\fi\fi
}
\makeatother

\usepackage[pagestyles,raggedright]{titlesec}
\usepackage{titletoc}

\usepackage{etoolbox}

\makeatletter
\patchcmd{\ttlh@hang}{\parindent\z@}{\parindent\z@\leavevmode}{}{}
\patchcmd{\ttlh@hang}{\noindent}{}{}{}
\makeatother

\newcommand{\secmark}{}
\newcommand{\marktotoc}[1]{\renewcommand{\secmark}{#1}}

\makeatletter
\renewcommand{\@dotsep}{1000}

\titleformat*{\section}{\center\large\bfseries}

\titleformat*{\subsection}{\center\bfseries}

\definecolor{airforceblue}{rgb}{0.36, 0.54, 0.66}
\definecolor{antiquefuchsia}{rgb}{0.57, 0.36, 0.51}
\definecolor{blush}{rgb}{0.87, 0.36, 0.51}
\definecolor{bondiblue}{rgb}{0.0, 0.58, 0.71}
\definecolor{MyGreen}{rgb}{0.0,0.5,0}
\definecolor{MyDarkRed}{rgb}{0.7,0,0}
\definecolor{MyBlue}{rgb}{0.0,0.0,.8}
\definecolor{Green}{rgb}{0.4,.8,0}
\newtheorem{theorem}{Theorem}

\def\be#1\ee{\begin{align}#1\end{align}}
\def\bsub#1\esub{\begin{subequations}#1\end{subequations}}
\def\bg#1\eg{\begin{gather}#1\end{gather}}
\def\ba{\begin{eqnarray}}
\def\ea{\end{eqnarray}}

\def\q{\qquad}
\def\f{\frac}
\def\df{\dfrac}
\def\dd{\text{d}}
\def\eps{\varepsilon}

\def\lb{\big\lbrace}
\def\rb{\big\rbrace}

\def\rm#1{\mathrm{#1}}

\def\de{\mathrm{d}}

\def\pe{\phantom{=\ }}

\newcommand{\R}{{\mathbb R}}
\newcommand{\Z}{{\mathbb Z}}

\newcommand{\cA}{{\mathcal A}}

\newcommand{\cD}{{\mathcal D}}

\newcommand{\cH}{{\mathcal H}}

\newcommand{\cL}{{\mathcal L}}

\newcommand{\cO}{{\mathcal O}}
\newcommand{\cP}{{\mathcal P}}
\newcommand{\cQ}{{\mathcal Q}}
\newcommand{\cR}{{\mathcal R}}
\newcommand{\cS}{{\mathcal S}}
\newcommand{\cT}{{\mathcal T}}

\newcommand{\cV}{{\mathcal V}}

\newcommand{\SL}{\mathrm{SL}}

\renewcommand{\sl}{{\mathfrak{sl}}}
\newcommand{\so}{{\mathfrak{so}}}
\newcommand{\bms}{{\mathfrak{bms}}}
\newcommand{\iso}{{\mathfrak{iso}}}

\renewcommand{\c}{{\mathfrak c}}
\newcommand{\h}{{\mathfrak h}}
\renewcommand{\theequation}{\thesection.\arabic{equation}}\numberwithin{equation}{section}

\begin{document}

\title{\Large{\textbf{\sffamily Dynamical symmetries of homogeneous minisuperspace models}}}
\author{\sffamily Marc Geiller, Etera R. Livine, Francesco Sartini}
\date{\small{\textit{ENS de Lyon, CNRS, Laboratoire de Physique, F-69342 Lyon, France}}}

\maketitle

\begin{abstract}

We investigate the phase space symmetries and conserved charges of homogeneous gravitational minisuperspaces. These (0\,+\,1)-dimensional reductions of general relativity are defined by spacetime metrics in which the dynamical variables depend only on a time coordinate, and are formulated as mechanical systems with a non-trivial field space metric (or supermetric) and effective potential. We show how to extract conserved charges for those minisuperspaces from the homothetic Killing vectors of the field space metric. In the case of two-dimensional field spaces, we exhibit a universal 8-dimensional symmetry algebra $\mathcal{A}=\big(\mathfrak{sl}(2,\mathbb{R})\oplus\mathbb{R}\big)\loplus\mathfrak{h}_2$, based on the two-dimensional Heisenberg algebra $\mathfrak{h}_2\simeq\mathbb{R}^4$. We apply this to the systematic study of the Bianchi models for homogeneous cosmology. This extends previous results on the $\mathfrak{sl}(2,\mathbb{R})$ algebra for  Friedmann-Lemaitre-Robertson-Walker  cosmology, and the Poincar\'e symmetry for Kantowski--Sachs metrics describing the black hole interior. The presence of this rich symmetry structure already in minisuperspace models opens new doors towards quantization and the study of solution generating mechanisms.

\end{abstract}

\thispagestyle{empty}
\newpage
\setcounter{page}{1}

\hrule
\vspace{-0.3cm}
\tableofcontents
\addtocontents{toc}{\protect\setcounter{tocdepth}{3}} 
\vspace{0.5cm}
\hrule

\newpage

\section{Introduction}

One of the most useful and conceptually profound results in modern physics comes from the work of Noether  \cite{Noether:1918zz,leone2018wonderfulness}. It relates symmetries, which are maps between classical solutions, to conservation laws, which are constants of the motion along each classical solution. In gauge theories, and in particular general relativity, Noether's two theorems together imply that symmetry transformations (which can be exact isometries or more generally asymptotic symmetries) give rise to codimension-2 surface charges (see \cite{Iyer:1994ys,Barnich:2007bf,Avery:2015rga,Compere:2018aar,DeHaro:2021gdv} for modern references). These encode important physical information about e.g. the mass (or energy), angular momentum, and possible radiation for certain classes of spacetimes. The existence in gauge theories of such non-vanishing charges signals the fact that boundaries turn otherwise gauge symmetries into physical symmetries. The charges generating these physical symmetries endow moreover the boundary with a non-trivial algebraic structure, which in gravity has been studied extensively since the pioneering work of Bondi, van der Burg, Metzner, and Sachs \cite{Bondi:1960jsa,Bondi:1962px,Sachs:1962zza,Sachs:1962wk,Barnich:2010eb,Flanagan:2015pxa,Madler:2016xju,Compere:2020lrt,Ruzziconi:2020cjt,Freidel:2021yqe,Fiorucci:2021pha}. In order to understand this fine structure of general relativity many authors have also turned to the study of symmetry charges and algebras in lower-dimensional gravity (see e.g. \cite{Cadoni:1999ja,Navarro-Salas:1999zer,Grumiller:2002nm,Maldacena:2016upp,Blommaert:2018iqz,Afshar:2019axx,Godet:2020xpk,ruzziconi2020conservation,Adami:2020ugu,Mertens:2020hbs} and \cite{Brown:1986nw,Ashtekar:1996cd,Barnich:2006av,Barnich:2012aw,Grumiller:2016pqb,Grumiller:2017sjh,Afshar:2016uax,Alessio:2020ioh,Geiller:2020okp,Geiller:2021vpg} in two and three spacetime dimensions respectively).

Going lower in dimension, one finds (0\,+\,1)-dimensional models, in which case field theories reduces to mechanical systems where the dynamical variables evolve only in time. In general relativity, such models are called minisuperspaces and arise from the reduction to a finite number of homogeneous degrees of freedom \cite{Ryan:1975jw}. This includes for example FLRW cosmologies, Bianchi models, and also the so-called Kantowski--Sachs cosmological models \cite{Kantowski:1966te} describing in particular the black hole interior (where due to the exchange of the radial and temporal coordinates the metric does indeed become homogeneous). This article is devoted to the study of the symmetry properties of such minisuperspace models. As expected, they inherit from full general relativity a remnant of diffeomorphism-invariance, which is the freedom in performing redefinitions of the time coordinate. The generator of these reparametrizations is the Hamiltonian itself. Our goal is to explain and show that in most minisuperspace models the Hamiltonian can in fact be embedded in larger symmetry algebras, which moreover bear similarities with the boundary symmetry algebras arising in four-dimensional general relativity. Although the minisuperspace models are not gauge field theories and have no spatial boundaries, the symmetries which we unravel act as physical symmetries (changing e.g. the mass of a black hole), and also interplay with the spatial cut-offs required in order to meaningfully define the homogeneous reduction of the Einstein--Hilbert action (and which are therefore remnants of  information about the boundary) \cite{francesco-thesis}.

The systematic study of such symmetries of minisuperspace models was initiated in \cite{BenAchour:2017qpb,BenAchour:2019ywl,BenAchour:2018jwq,BenAchour:2019ufa,BenAchour:2020ewm,BenAchour:2020njq} (see also  \cite{Pioline:2002qz}), where it was shown in the context of FLRW cosmology that the Hamiltonian belongs to an $\sl(2,\R)$ algebra called CVH. This latter is spanned by the Hamiltonian $H$ together with the volume $V$ and the generator $C$ of isotropic dilatations on phase space. This was extended in \cite{Achour:2021lqq} to include spatial curvature and a cosmological constant. In \cite{Geiller:2020xze}, the present authors have studied the Kantowski--Sachs model describing the black hole interior, and found that there the CVH algebra gets extended to an $\iso(2,1)$ algebra of conserved charges encoding the dynamics of the system in an algebraic manner. This was then extended to Schwarzschild (A)dS metrics in \cite{Achour:2021dtj}. This body of results led to a completely symmetry-based quantization of the black hole interior model in  \cite{Sartini:2021ktb}. In \cite{BenAchour:2020xif,Geiller:2021jmg} (see also \cite{Pailas:2020xhh} for related results), the infinite-dimensional enhancement of $\sl(2,\R)$ and $\iso(2,1)$ to the Virasoro and $\bms_3$ algebras respectively has been studied, and it was shown that these extended transformations (which are of course motivated from three and four-dimensional gravity) enable to generate e.g. a cosmological constant term in the minisuperspace models. It was also shown that the (0\,+\,1)-dimensional minisuperspace actions can be written as geometric actions for the enhanced symmetry groups.

Although minisuperspace models seem to be very simple, they contain a rich and physically-useful symmetry structure whose origin is not yet understood in a systematic manner. This is what we aim at improving here. At first sight, one could think of approaching this problem using the Hamiltonian formulation, where roughly speaking the requirement of homogeneity amounts to solving the vector constraint. We are then left with the scalar constraint only, which furthermore splits as $\cH=\cH_\text{kin}+\cH_\text{pot}$ into a kinetic term and a potential term. The kinetic term depends on the ADM momenta while the potential only depends on the Ricci scalar of the three-dimensional slice (and the spatial volume in the presence of a cosmological constant). As explained in appendix \ref{ADM}, the kinetic term $\cH_\text{kin}$ always forms an $\sl(2,\R)$ algebra together with the volume (which is the integral on the slice of $\sqrt{q}$) and the smeared extrinsic curvature. This is however not sufficient in order to obtain symmetries of the model since (as we will recall below) the construction of the conserved charges (which generate the said symmetries) requires the \textit{full} Hamiltonian to be part of the algebra, and not only its kinetic part. When the minisuperspace model is such that the three-dimensional Ricci curvature vanishes, the Hamiltonian has only a kinetic term and we are therefore guaranteed that the CVH algebra exists. This is what happens for example in flat FLRW cosmology. In more general models, where a non-trivial potential is present, we are therefore left with the question of how to recover (if possible) a CVH-type algebra. It turns out that studying this problem frontally by computing Poisson brackets involving the three-dimensional Ricci scalar is too cumbersome. An alternative route, which is the one we set out to study here, is to use the superspace formulation.

The authors \cite{Christodoulakis:2013sya,Christodoulakis:2013xha,Christodoulakis:2012eg,Dimakis:2015rba,Terzis:2015mua,Dimakis:2016mpg,Christodoulakis:2018swq} have also studied the phase space symmetries of homogeneous cosmological models (without explicit reference to an $\sl(2,\R)$ CVH algebra however), but using a field space formulation. This formulation is based on the following simple observation: When evaluating the Einstein--Hilbert action on a minisuperspace line element, the action reduces to a one-dimensional mechanical model of the form
\be
\cS=\int\de t\,\left(\f{1}{2}g_{\mu\nu}\dot{q}^\mu\dot{q}^\nu-U(q)\right).
\ee
For simplicity and illustrative purposes, we have set here the lapse inherited from the line element to $N=1$. Its choice (in a field-dependent manner) will however be fundamental for the construction, as we explain in details in the rest of the article. Here $g_{\mu\nu}$ is \textit{not} the spacetime metric, but rather the metric on field space, or superspace. It controls the form of the kinetic term which depends on the coordinates $q^\mu$ on field space. $N$ is the lapse inherited from the line element, and $U$ is the potential inherited from the spatial integral of the three-dimensional Ricci scalar (i.e. it is essentially the potential $\cH_\text{pot}$ appearing in the Hamiltonian). As noticed in \cite{Christodoulakis:2013sya,Christodoulakis:2013xha,Christodoulakis:2012eg,Dimakis:2015rba,Terzis:2015mua,Dimakis:2016mpg,Christodoulakis:2018swq}, the symmetry properties of the minisuperspace model can be inferred from the symmetries of the supermetric and the potential. To explain how this comes about, suppose we identify a vector field such that $\pounds_\xi g_{\mu\nu}=\lambda g_{\mu\nu}$ and $\pounds_\xi U=-\lambda U$ with $\lambda=\text{constant}$. Then the quantity $C=\xi^\mu p_\mu$ obtained by contracting the vector field with the momentum is such that its Poisson bracket with the Hamiltonian gives $\lb C,H\rb=\lambda H$. From this one can immediately conclude that $\cQ=C-t\lambda H$ is a conserved quantity, which therefore generated a symmetry of the theory. This is in essence the relationship between the field space geometry $g_{\mu\nu}$ and the symmetries of the minisuperspace phase space. In the CVH algebra, the volume $V$ gives rise to a conserved quantity which is quadratic in time, at the difference with $C$ which is linear. Our goal is to generalize the construction of \cite{Christodoulakis:2013sya,Christodoulakis:2013xha,Christodoulakis:2012eg,Dimakis:2015rba,Terzis:2015mua,Dimakis:2016mpg,Christodoulakis:2018swq} outlined above and explain how such quadratic charges (and therefore a CVH algebra) can be constructed. As we are going to see in details, this depends heavily on the choice of the lapse and on the potential $U$, which is what makes the analysis involved.

In summary, we show that the above construction based on conformal Killing symmetries of the supermetric $g_{\mu\nu}$ can, for some minisuperspace models, identify many more conserved charges than that forming the CVH algebra. For FLRW cosmology with a scalar field, Kantowski--Sachs cosmologies and the Bianchi models (III, V, VI), we exhibit an 8-dimensional symmetry algebra $\cA=\big(\sl(2,\R) \oplus\R\big)\loplus\mathfrak{h}_2$, where $\mathfrak{h}_2\simeq\R^4$ is the two-dimensional Heisenberg algebra. For the Bianchi models (VIII, IX), only the $\sl(2,\R)$ CVH subalgebra survives; while, oddly enough, for the Bianchi models (IV, VII), we find that the construction fails and does not lead to an algebra. These are all models with a two-dimensional field space (spanned essentially by two scale factors, or one scale factor and the scalar field in the case of FLRW with matter). We conclude the paper with a study of three-dimensional field spaces, which arise in the Bianchi I and II models. In the case of Bianchi I, we find an algebra which has dimension 30, while for Bianchi II the algebra is 4-dimensional. In both cases, the $\sl(2,\R)$ CVH algebra appears as a subalgebra. The reason for which some of these algebras are much bigger than the dimensionality of the phase space is that the generators satisfy dependency relations. In the case of the algebra $\cA$ these are analogous to Sugawara constructions starting from four building blocks which are the generators of the two-dimensional Heisenberg algebra $\h_2$. The choice of a preferred subalgebra on which to base the group quantization will therefore depend on which physical variable we wish to represent in the quantum theory \cite{francesco-thesis,Sartini:2021ktb}. It is therefore desirable to obtain these enlargements of the $\sl(2,\R)$ symmetry algebra, even if classically these may seem a redundant description of the dynamics and the symmetries. We note that this article is devoted to the presentation of the technical construction of these algebras, and to their classification in the case of FLRW, Kantowski--Sachs, and Bianchi models. Physical applications (e.g. to quantization and solution generating transformations) will be presented elsewhere.

This article is organized as follows.
We start in section \ref{sec:generalities} by discussing in general about  homogeneous reductions of general relativity. This shows how minisuperspace models are formulated as mechanical Lagrangians of the form \eqref{mini_lagrang} with a given field space geometry. We then explain the relationship between  conformal Killing vectors of this field space geometry and conserved quantities.
We clarify the role of the potential and of the lapse in this procedure. Section \ref{sec:2d} is then devoted to the study of models which give rise to a two-dimensional field space. We explain how the algebra $\cA$ arises from the geometry of two-dimensional Minskowski field space, and how the presence of a non-vanishing potential may actually restrict this algebra to a subalgebra. We then apply these criteria to the Bianchi models and to the black hole interior. Finally, section \ref{sec:Bianchi I II} presents the example of three-dimensional field spaces, which arise in the case of the Bianchi I and II models (which have three scale factors). Appendix \ref{Bianchi_metrics} presents a useful summary of the properties of the various Bianchi models, listing their line elements, potentials, and fiducial volumes.

\section{Minisuperspaces}
\label{sec:generalities}

In this section, we present the formalism used to describe homogeneous minisuperspace models in general relativity and to study their symmetries. We start by reviewing the homogeneous reduction of gravity, and introduce a triad formulation which allows to separate the homogeneous dynamical fields from the background geometric structure of the various minisuperspace models. This enables us to identify an internal metric on field space (or superspace), whose geometry controls the dynamics of the homogeneous symmetry-reduced action. We then explain how the symmetries of this field space geometry can be translated into algebraic and symmetry structures on the phase space of the minisuperspace models.

\subsection{Homogeneous reduction of gravity}
\label{sec:homogeneous}

In this work, we are interested in the study of homogeneous minisuperspace models, for which the Einstein--Hilbert action reduces to a mechanical action integrated over a single time variable. In order to set the stage, let us consider a four-dimensional spacetime equipped with a metric of the general form
\be\label{minisuperspace}
\de s_{4\rm{D}}^2 = -N(t)^2 \de t^2 + q_{\alpha\beta}(x,t) \de x^\alpha \de x^\beta\,, \q\q
q_{\alpha\beta}(x,t)= e_\alpha^i(x)\, e_\beta^j(x)\, \gamma_{ij}(t)\,.
\ee
In this parametrization the spatial triad $e^i_\alpha$ does not depend on the time coordinate, but only on the position on the spatial slice. All the dynamical information is therefore captured in the spatially homogeneous but \textit{time-dependent} fields $\gamma_{ij}$. This parametrization is a convenient way to keep track of the fact that in e.g. Bianchi models (see appendix \ref{Bianchi_metrics}) the metric can depend explicitly on the spatial coordinates but the dynamical fields which evolve in time are homogeneous. In the full ADM field theory the temporal and spatial dependencies cannot be disentangled and are captured by a single field $q_{\alpha\beta}(t,x)$, which is the metric on a slice. Here, however, because of homogeneity the gravitational action will reduce to a mechanical model with only evolution in time.

In minisuperspace models, a leftover of the ADM formulation is encoded in the fact that the Hamiltonian of the mechanical model is the scalar constraint of general relativity. Because the line element \eqref{minisuperspace} does not contain shift terms, i.e. space/time cross  terms proportional to $\de t\,\de x^\alpha$, we need however to ensure for consistency that the ADM vector constraints are satisfied (which amounts to requiring that the projection of the Einstein equations on the spatial slice are satisfied). To this purpose, let us first compute the extrinsic curvature of the three-dimensional slice,
\be
K_{\alpha\beta} = \f{1}{2N} \dot q_{\alpha\beta} = \f{1}{2N} e^i_\alpha e^j_\beta\, \dot{\gamma}_{ij}\,,\q\q  K\coloneqq q^{\alpha\beta} K_{\alpha\beta} = \f{1}{2N} \gamma^{ij} \,\dot \gamma_{ij}\,,
\ee
where $\gamma^{ij}$ is the inverse of the internal metric, $\gamma^{ij}\gamma_{jk}=\delta^{i}{}_{k}$.

Writing $D_\alpha$ for the three-dimensional covariant derivative\footnote{Using the convention $D_\alpha$ for the 3d covariant derivative, we keep the notation $\nabla_\mu$ for the covariant derivative with respect to the  metric $g_{\mu\nu}$ on field space.} with respect to the hypersurface metric $q_{\alpha\beta}$, imposing  the vector constraint  amounts  to the requirement that
\be
D_\alpha \big(K q^{\alpha\beta}- K^{\alpha\beta}\big) =\f{1}{2N}\Big(\big(\gamma^{k\ell} \dot \gamma_{k\ell}\big)\gamma^{ij} + \dot \gamma^{ij}\Big) D_\alpha \big ( e^\alpha_{i} \, e^\beta_j\big )\stackrel{!}{=}0\,. \label{vector_constr}
\ee
The internal metric $\gamma_{ij}$ must be chosen such that this equation is satisfied (if this is not already the case given the form of $e^\alpha_i$). As we will see shortly, it turns out that in the majority of the cases relevant for general relativity the internal metric only depends on two dynamical fields, which will be denoted by $a(t)$ and $b(t)$.

When the constraint \eqref{vector_constr} is satisfied, the dynamical evolution is described by the scalar constraint alone, and the evolution equations derived from the symmetry-reduced homogeneous action coincide with the homogeneous reduction of the full Einstein field equations. The Einstein--Hilbert action for the line element \eqref{minisuperspace} reduces to
\be\label{Einstein_mini}
&\phantom{=\ \;}\f{1}{16 \pi G}\int_{t_\rm f}^{t_\rm i} \int_\Sigma\de^4 x\,\sqrt{-g}\,\big(\cR-2\Lambda\big)\cr
&=\f{1}{16 \pi G} \int_{t_\rm f}^{t_\rm i} \int_\Sigma \de^3 x \,N\sqrt{q}\left (K^2-K_{\alpha\beta}K^{\alpha\beta} + R^{(3)} - 2 \Lambda \right )  +\f{1}{8\pi G} \left .\int_\Sigma \de^3 x\, \sqrt{q}\, K\right |_{t_\rm{i}}^{t_\rm{f}} \notag \\
&= \f{\cV_0}{G} \int_{t_\rm f}^{t_\rm i}\de t\,\sqrt{\gamma}\left[\f{1}{4N}\Big((\gamma^{ij} \dot \gamma_{ij})^2+ \dot \gamma_{ij} \dot \gamma^{ij}\Big) - 2 N \Lambda \right ] +\f{1}{16 \pi G} \int_{t_\rm f}^{t_\rm i} \int_\Sigma \de^3 x \,N\sqrt{q}\,R^{(3)}  +\cS_{\rm{GHY}}\,.\q
\ee
Here we have introduced the fiducial volume
\be
\cV_0 \coloneqq\f{1}{16\pi} \int_\Sigma \de^3 x \,|e|\,,
\ee
where $e= \det(e^i_\alpha)$, which we want to be finite. Depending on the topology of $\Sigma$ this may require the introduction of cut-offs (i.e. further fiducial quantities) on the spatial integrations, as done in appendix \ref{Bianchi_metrics} for the various models of interest. Finally, the Gibbons--Hawking--York (GHY) term is a total derivative and does not play any role in the evolution of the classical system.

This symmetry-reduced action \eqref{Einstein_mini} should be seen as the action for a mechanical model, describing the evolution in time of the degrees of freedom $\gamma_{ij}(t)$. The kinetic term comes from the extrinsic curvature contribution to the four-dimensional curvature in the Einstein--Hilbert Lagrangian, while the three-dimensional Ricci scalar $R^{(3)}$ plays the role of the potential. Unfortunately there is no general formula at this stage for splitting  $R^{(3)}$ into a time-dependent dynamical part and a purely spatial frame-dependent piece (see appendix \ref{triad_decomp} for the explicit computation of the 3d curvature), and its expression will be evaluated on a case by case basis. The equation of motion obtained from the variation of this action with respect to the lapse imposes, as expected, the scalar constraint. 

Let us now illustrate this general construction with the example of the Bianchi III line element (we give all the other minisuperspace examples in section \ref{application})
\be\label{exmpl:metric}
\de s^2_{\rm{III}} &= -N(t)^2 \de t^2 + a(t)^2 \de x^2 + b(t)^2 L_s^2 \left (\de y^2+\sinh^2 y\, \de \phi^2\right )\,.
\ee
For this line element the triad and the internal time-dependent metric are given by
\be
e^1_\alpha\de x^\alpha = \de x\,,\q e^2_\alpha\de x^\alpha = L_s\, \de y \,,\q e^3_\alpha\de x^\alpha = L_s\, \sinh y\, \de \phi \,,\q\gamma_{ij} =\rm{diag} \left (a^2, b^2,b^2\right )\,.
\ee
One can easily check that the condition \eqref{vector_constr} is indeed satisfied. The length scale $L_s$ has been introduced in order to have dimensionless fields $a$ and $b$. Since the spatial slice has the topology $\R^2 \times S^1$, we need to introduce two fiducial scales in order to have a finite spatial volume. Introducing cut-offs\footnote{Note that there are several fiducial scales entering the equations. $L_s$ is used to ensure that the dynamical fields are dimensionless, while $L_x,L_y,\dots$ are dimensionful cut-offs (fiducial lengths) on the variables $x,y,\dots$, while $x_0,y_0,\dots$ are dimensionless cut-offs.} in the two non-compact $x$ and $y$ directions we get
\be
\cV_0 = \f{1}{16 \pi} \int_0^{L_x} \de x \int_0^{y_0}\de y \int_0^{2\pi} \de \phi\, L_s^2|\sinh y| = \f{1}{4}L_x L_s^2\sinh^2 \left (\f{y_0}{2}\right)\,.
\ee
Up to the total time derivative, the action \eqref{Einstein_mini} evaluated for \eqref{exmpl:metric} then gives
\be\label{action example}
\cS=-\f{\cV_0}{G}\int \de t \left [\f{1}{N}\big(4 b\, \dot a \dot b +2 a\, \dot b^2\big)+2N\left (\f{a}{L_s^2}+\Lambda a b^2\right )\right ]\,.
\ee
This is indeed a mechanical system encoding the evolution in time of the two degrees of freedom $a$ and $b$, which are the dynamical components of the minisuperspace line element.
Our goal is to now describe the phase space symmetries of such minisuperspace models, that is identify the (possibly) time-dependent constants of motion, to be interpreted as Noether charges generating the symmetries of the system. For this, we turn to a field space formulation.

\subsection{Field space formulation}

The mechanical actions of the form \eqref{action example}, obtained from a homogeneous reduction of general relativity, always describe a particle moving in an auxiliary curved spacetime and subject to a potential. The coordinates on this spacetime are the independent fields which appear in the internal metric $\gamma_{ij}$. This dynamics is described by mechanical Lagrangians of the form
\be\label{mini_lagrang}
\cL = \f{1}{2N} \tilde g_{\mu\nu} \dot{q}^\mu \dot{q}^\nu - N \tilde{U}(q)\,.
\ee
The tensor $\tilde g_{\mu\nu}(q)$ is the invertible metric on the field space, on which the $q$'s are coordinates. We will denote by $\de \tilde s^2_\rm{mini}$ the corresponding line element. These quantities should not be confused with the original spacetime metric and coordinates\footnote{We still use Greek letters to denote the coordinates on the field space. In summary, we have three kinds of indices:
\begin{itemize}
\item The coordinates and the metric on the hypersurface are $x^\alpha$ and $q_{\alpha\beta}$, where $\alpha,\beta,\gamma,\dots\,\!=1,\dots,3$.
\item The internal indices for the internal metric and the frame are $i,j,k,\dots\,\!=1,\dots,3$. We also use Latin letters to label the Killing vectors below, but place them between parentheses: $(i)$, $(j)$, $\dots$.
\item The coordinates and the metric on field space are $q^\mu$ and $g_{\mu\nu}$, but the range can vary on a model-dependent basis.
\end{itemize}}.
In particular, note that the dimension of the field space metric can vary on a model-dependent basis even when the dimension of space-time is fixed.

In what follows it will sometimes be useful to perform redefinitions of the fields entering the minisuperspace model \eqref{minisuperspace}. This induces a redefinition of the coordinates $q^\mu$ in the field space and in turn a redefinition of the field space metric. It will for example be convenient to write the flat field space metrics in explicit Minkowski form. These changes of coordinates simply exploit the freedom in changing variables in the minisuperspace model. For example, in FLRW cosmology one can work with the scale factor $a$ or with the volume $a^3$. The range of the coordinates in the field space formulation simply follows the range of the dynamical variables in the minisuperspace model. The equivalence between the equations of motion derived from the homogeneous Lagrangian \eqref{mini_lagrang} and the homogeneous reduction of Einstein field equations, ensures the consistency of the field space approach. The presence of singularities or degenerate metrics occurs then in the same way as in the standard general relativity formulation.

In the example of the Bianchi III model described above, the Lagrangian corresponding to the reduced action \eqref{action example} is indeed of the form \eqref{mini_lagrang} with $q^\mu=(q^1,q^2)=(a,b)$ and
\be
\de \tilde s^2_{\rm{mini}} = -8 b\, \de a\, \de b - 4 a\, \de b^2\,, 
\q\q
\tilde{g}_{\mu\nu}=
-4\begin{pmatrix}
0&b\\
b&a
\end{pmatrix}
\,,
\q\q
\tilde{U} = 2 \left (\f{a}{L_s^2}+\Lambda a b^2\right )\,.
\ee
The lapse function $N$ is a remnant of the diffeomorphism invariance inherited from the full theory. It plays the role of a Lagrange multiplier enforcing the Hamiltonian constraint, and it ensures that the model is invariant under time reparametrizations. The Lagrangian \eqref{mini_lagrang} is indeed invariant under the symmetry
\be
\delta_f q^\mu = f \dot{q}^\mu\,,\q\q \delta_f N = f\dot{N}+ \dot{f} N\,, 
\ee
for a generic function $f$ which can also be field-dependent. Performing a Legendre transform of $\eqref{mini_lagrang}$ we obtain the Hamiltonian on the phase space $\{q^\mu,p_\nu\} =\delta^\mu_\nu$, which is
\be
\cL =p_\mu\dot{q}^\mu -N \cH\,,\q\q p_\mu=\f{1}{N}\tilde{g}_{\mu\nu}\dot{q}^\nu \,,\q\q H\coloneqq N\cH = \f{1}{2}N \tilde{g}^{\mu \nu} p_\mu p_\nu + N\tilde{U}(q)\,,
\label{curved_particle_model}
\ee
with $\cH\approx 0$ on-shell. At the classical level any choice of lapse leads to an equivalent description of the dynamics. It is indeed common to simply view a choice of lapse, which can also be field-dependent, as a choice of clock. For example, working with the proper time $\tau$, defined by $\de\tau=N \de t$, is equivalent to setting $N=1$.

As we are about to see, in order to study the symmetries it will however be important to choose the lapse as a field-dependent function $N(q)$, and to keep track of its interplay with the field space metric and the potential. For this reason, we will stick to the time variable $t$ and denote $\de/\de t$ by a dot, and view the lapse as a conformal factor in front of the field space metric. We therefore introduce a rescaled metric such that
\be
g_{\mu\nu} \coloneqq \f{1}{N} \tilde g_{\mu\nu}\,,\q\q p^\mu =g^{\mu\nu} p_\nu = \dot q^\mu\,, \q\q \xi_\mu = g_{\mu\nu} \xi^\mu\,,\q\q
\partial_\mu\coloneqq\f{\partial}{\partial q^\mu}\,,
\ee
for any vector on the field space $\xi^\mu$. The evolution of a phase space observable $\cO$ with respect to $t$ is given by
\be
\dot \cO\coloneqq\f{\de \cO}{\de t} = \partial_t \cO + \lb \cO, H\rb\,,\q\q H\coloneqq N \cH =\f{1}{2}p^\mu p_\mu +N\tilde{U}\,.  
\ee
Below we will interpret a change of lapse as a change of field space metric instead of a change of time (although the two pictures are equivalent at the end of the day). We would like to remark that the phase space is actually independent on the choice of lapse, and hereafter we will think it as a function on phase space, instead of a free field.

When working with the rescaled metric $g_{\mu\nu}$, the equation of motion obtained by varying \eqref{mini_lagrang} with respect to $q^\mu$ is
\be
\ddot{q}^\mu+\Gamma^\mu_{\rho\sigma}\dot{q}^\rho\dot{q}^\sigma+g^{\mu\nu}\partial_\nu\tilde{U}=0\,,
\ee
where $\Gamma$ is the Christoffel connection of the field space metric $g_{\mu\nu}$. In the \textit{free} case, i.e. without potential, this is simply the geodesic equation, and the scalar constraint $H =p^\mu p_\mu/2 \approx 0$ restricts to the null geodesics. In this case the evolution of a phase space functions $\cO$ along any trajectory is given by
\be
\dot \cO = \partial_t \cO + p^\mu \partial_\mu \cO\,, \q\q p^\mu \partial_\mu p^\nu =0\,.
\ee
In order to unravel the symmetries on the phase space of this system, we first need to study the Killing vectors of the field space metric $g_{\mu\nu}$, and to understand their relation with conserved quantities.

\subsubsection{Killing vectors and conserved quantities}

We now turn to the core of the discussion, and explain how the phase space symmetry algebra is related to the symmetries of the field space metric. We begin with the analysis of the \textit{free case}, for a vanishing potential. Let us start by considering the initial field space metric $\tilde{g}_{\mu\nu}$ without conformal rescaling by the lapse. We are interested in its conformal Killing vectors $\xi^\mu$, which by definition are such that
\be 
\pounds_\xi \tilde g_{\mu\nu} = \tilde \nabla_\mu \xi_\nu + \tilde \nabla_\nu \xi_\mu \doteq \tilde \lambda \tilde g_{\mu\nu}\,,\q\q   \pounds_\xi \tilde{g}^{\mu\nu} =-\tilde{\lambda} \tilde{g}^{\mu\nu} \,.
\ee
Here $\tilde \nabla$ is the covariant  derivative with respect to $\tilde g$. Whenever $\tilde \lambda$ is constant and non-vanishing we refer to these \textit{conformal} Killing vectors as \textit{homothetic} Killing vectors, while in the vanishing case we have \textit{true} Killing vectors. It is now easy to see that the presence of a conformal factor $N$ does not change the fact that a given vector $\xi^\mu$ is conformal, but changes however the value of $\lambda$. Indeed, we have
\be\label{rescaling in Killing eq}
\pounds_\xi g_{\mu\nu} =  \nabla_\mu \xi_\nu +  \nabla_\nu \xi_\mu =  \pounds_\xi \left (\f{1}{N} \tilde g_{\mu\nu}\right ) = \big (\tilde \lambda - \pounds_\xi\log N  \big ) \left (\f{1}{N} \tilde  g_{\mu\nu}\right ) \eqqcolon \lambda {g}_{\mu\nu}.
\ee
This means that a field-dependent change of lapse can change the nature of a conformal Killing vector, by making the weight non-constant, constant, or vanishing. This is a first indication of the importance of the lapse in this analysis.

In order to set the stage and continue the discussion, let us now assume that we have fixed the lapse once and for all, and that we have found the homothetic Killing vectors of the metric $g_{\mu\nu}$, i.e. the vectors such that $\lambda$ is constant (and possibly vanishing). The whole set $\left \{ \xi_{(i)} \right\}$ of homothetic and true Killing vectors then forms an algebra with commutation relations
\be
\big[\xi_{(i)}, \xi_{(j)}\big] = \left (\xi_{(i)}^\mu \nabla_\mu \xi_{(j)}^\nu  -\xi_{(j)}^\mu \nabla_\mu \xi_{(i)}^\nu \right )\partial_\nu = {c_{ij}}^k \xi_{(k)}\,,
\ee
where the coefficients ${c_{ij}}^k$ are the structure constants. Furthermore, as shown in appendix \ref{HKV_prop}, these Killing vectors are also solutions of the geodesic deviation equation
\be
p^\mu p^\nu \nabla_\mu \nabla_\nu \xi_\rho = - R_{\rho\mu\sigma\nu} p^\mu p^\nu \xi^\sigma \,.
\ee
With these vectors at hand let us now consider the phase space functions
\be\label{phase space C}
C_{(i)} \coloneqq p_\mu \xi_{(i)}^\mu\,.
\ee
By construction we then get that
\be
\lb C_{(i)},H\rb = p^\nu \partial_\nu \big ( p_\mu \xi_{(i)}^\mu \big ) = \f{1}{2}\lambda_{(i)}  g_{\mu \nu} p^\mu p^\nu =\lambda_{(i)} H\,,
\ee
which immediately implies that we have conserved quantities defined as
\be
\cQ_{(i)} \coloneqq  C_{(i)} - t \lambda_{(i)} H\,,\q\q \dot \cQ_{(i)}=\partial_t\cQ_{(i)}+\lb\cQ_{(i)},H\rb =0\,.
\label{linear_charges}
\ee
These conserved charges are linear in time, they are evolving constants of motion, as uncovered for FLRW cosmology and black hole mechanics in the previous works \cite{BenAchour:2017qpb,BenAchour:2019ufa,Geiller:2020xze,BenAchour:2020njq,Geiller:2021jmg,Sartini:2021ktb}. The charge $\cQ_{(i)} $ can indeed be understood as the initial condition at $t=0$ for the observable $C_{(i)}$. This set of conserved charges form part of the phase space symmetry algebra which we set out to unravel. One can already add to this set the Hamiltonian at this stage. Then the question which  remains is whether we can extend the construction to include charges which are quadratic (or higher polynomials) in time. This will be the goal of the next section.

As conserved quantities, the charges \eqref{linear_charges} are generators of symmetries of \eqref{mini_lagrang} whose infinitesimal action is
\bsub
\be
\delta q^\mu = \{q^\mu, \cQ_i\} &= \xi_{(i)}^\mu - \lambda_{(i)}  t\, p^\mu = \xi_{(i)}^\mu -t \lambda_{(i)}  \dot q^\mu \,,\\
\delta p_\mu = \{p_\mu, \cQ_i\} &= -p_\nu \partial_\mu \xi_{(i)}^\nu+\f{1}{2}t\lambda_{(i)}   p_\nu p_\rho \partial_\mu g^{\nu\rho}\,.
\ee
\esub
The interpretation of this symmetry transformation is as follows. Let us start with a geodesic which is affinely parametrized. Since the conformal Killing vectors are solutions of the geodesic deviation equation, the result of the transformation is to first move the point $q^\mu$ to a nearby geodesic. Then, is shifts along this new geodesic by $- t \lambda_{(i)} p^\mu$, in order to account for a potential dilatation of the spacetime generated by $\xi_{(i)}$, so as to obtain again an affine parametrization of the new solution.

Using the Poisson bracket on phase space, we can finally show that the conserved charges constructed from the $\xi_{(i)}$'s form an algebra which reproduces the Lie algebra of the Killing vectors, i.e.
\be\label{linear algebra}
\lb\cQ_{(i)},\cQ_{(j)}\rb = -{c_{ij}}^k\cQ_{(k)}\,.
\ee
As we will illustrate below, it should be noted that the size of this algebra (or even its existence) depends on the choice of the lapse (or conformal factor for the metric) $N$. This doesn't mean that we cannot recover the corresponding conserved charges in any time parametrization, but conversely that their form might depend on the history of the system. This happens because the factors of $t$ that appear in the expressions of the conserved charges \eqref{linear_charges}, when expressed in terms of another time $\tilde t$, would lead to non-local terms given by the integral of the ratio of the two lapses, associated
with the original and the new times: $N \de t = \widetilde N \de \tilde t \Rightarrow t = \int  \de \tilde{t}\, ( \widetilde N/N)$  \cite{francesco-thesis}.

Inspired by the symmetry structure which has been exhibited in previous work for cosmological minisuperspaces \cite{BenAchour:2020njq,BenAchour:2020ewm,BenAchour:2019ufa,BenAchour:2019ywl,BenAchour:2020xif,Achour:2021lqq} and for the black hole interior \cite{Geiller:2021jmg,Geiller:2020xze,Achour:2021dtj}, we would like to extend the above construction to charges which are quadratic in time. This can be done in a systematic manner in the case of two-dimensional field spaces, which fortunately enough covers many homogeneous models of general relativity. The construction in the case of higher-dimensional field spaces has to be done on a case by case basis.

\subsubsection{Inclusion of a potential}
\label{sec2:potential}

Before turning to the study of the two-dimensional field space geometries, we present a first brief discussion concerning the inclusion of a non-vanishing potential. Once again, this is an aspect which is heavily impacted by the choice of lapse, and which we illustrate below with many examples. Let us assume that the lapse has been chosen to be field-dependent. 

If the product $N(q)\,\tilde{U}(q)$ is equal to a constant $U_0$, then this constant does not play any fundamental role form the point of view of the symmetry algebra. All the above discussion is still valid, except for the fact that previously null trajectories now become massive or tachyonic. Indeed, such a constant potential is a boundary term for the point of view of the action. In the charge algebra, the effect of a constant term $U_0$ is simply to shift the value of the Hamiltonian to obtain the charge
\be
\cQ_0 \coloneqq H -U_0\,,\q\q \{\cO,H\} =\{\cO,\cQ_0\}\,.
\ee
This is exactly what happens in the case of KS cosmologies and the black hole interior \cite{Geiller:2021jmg,Geiller:2020xze}. We will review the explicit examples below.

Now, if the potential contains a non-constant piece, we use the following notation
\be\label{new potential}
N(q)\,\tilde{U}(q) =U(q) + U_0
\ee
to disentangle between a possibly constant term $U_0$ and a field-dependent piece $U(q)$. This separation evidently depends on the choice of lapse, and the constant piece is not uniquely determined. However, we must add an extra condition on the Killing vectors in order to be able to define the charges \eqref{linear_charges}. This condition is the requirement that the non constant piece gets rescaled similarly to the supermetric, i.e.
\be\label{potential_condition}
\pounds_\xi U = \xi^\mu\partial_\mu U= -\lambda U\,.
\ee
Indeed, this is necessary in order to ensure that we still have
\be
\lb C_{(i)}, \cQ_0\rb = \lb C_{(i)}, \f{1}{2}p^\mu p_\mu+  U \rb = \f{1}{2}\lambda_{(i)} p^\mu p_\mu - \xi^\mu \partial_\mu U = \lambda_{(i)} \cQ_0\,,
\ee
so that we can once again define the conserved quantities
\be
\cQ_{(i)} \coloneqq  C_{(i)} - t \lambda_{(i)} \cQ_{0}\,,\q\q \dot \cQ_{(i)} =0\,,
\ee
which in turn generate symmetries and form an algebra. Moreover the condition \eqref{potential_condition}, for $\lambda\neq 0$, uniquely fixes the constant piece $U_0$, so that the separation \eqref{new potential} actually depends both on the lapse and the vector $\xi$.

It is important to notice that under a conformal rescaling $\tilde{g}_{\mu\nu}=Ng_{\mu\nu}$ of the metric, corresponding to a change of lapse, the initial requirement $\pounds_\xi U=-\lambda U$ implies
\be
\pounds_\xi \tilde g_{\mu\nu}  =  \big(\lambda +\pounds_\xi \log N\big) \tilde g_{\mu\nu} =\tilde{\lambda} \tilde g_{\mu\nu} \,,  
\q
\pounds_\xi  \tilde{U}   = -\big(\lambda +\pounds_\xi \log N\big)\tilde{U} + \lambda \f{U_0}{N} =-\tilde{\lambda}\tilde{U}+\lambda \f{U_0}{N}\,.
\ee
This means that in the presence of a non-trivial potential the desired properties on the transformation of the metric and the potential are not covariant under changes of the lapse. This makes the analysis of the symmetries and their algebra very involved, since both the number of homothetic Killing vectors (satisfying the homothetic conditions on both the metric and the potential) and the decomposition of the potential \eqref{potential_condition} depend on the lapse. This is why we will have to perform this analysis on a case by case basis.

Another method to deal with a non-trivial potential is to use the so-called Eisenhart lift \cite{Cariglia:2015bla}. This consists in extending the field space by the addition of a new field $y$, in order to map the $n$-dimensional problem with potential to an $(n+1)$-dimensional free system. The Eisenhart lift is explicitly given by an enlarged supermetric $G_{AB}$, with
\be
G_{AB}=
\begin{pmatrix}
g_{\mu\nu} & 0 \\
0 & 1/(2U)
\end{pmatrix}
\ee
and where the new coordinates (encoding the fields) are $x^A=(q^\mu, y)$. The null geodesic equation of $G_{AB}$ on the configuration space $(q^\mu, y)$ gives exactly the same equation as $H=p^\mu p_\mu + U$, as can easily be verified by computing the Hamiltonian equation for $H_\rm{lift}= G^{AB} p_A p_B$ with $p_A =(p_\mu,p_y)$. In particular, the equation of motion for $p_y$ states that it is a constant. Once we pull back this into the equations for $q^\mu$ and $p_\mu$, these become equivalent to the original equations for the $n$-dimensional problem.

It is possible to show that searching for the homothetic Killing vectors of the supermetric $G_{AB}$ of the enlarged space is the same as finding a simultaneous solution of the rescaling equations \eqref{potential_condition} and \eqref{rescaling in Killing eq} for the supermetric and potential, $\pounds_\xi g_{\mu\nu}  ={\lambda} g_{\mu\nu}$ and  $\pounds_\xi U  =-{\lambda} U$. Although this apparently looks equivalent, it nevertheless allows for vector fields $\xi$ which can depend on the new coordinate $y$ and have non-trivial components along $\partial_{y}$. Unfortunately, this new coordinate does not have a physical meaning and there is no way to extract non-trivial Dirac observables from such a homothetic Killing vector on the Eisenhart lift metric.

Nevertheless, there exists an improved version of the lift, called the Eisenhart--Duval lift \cite{Cariglia:2016oft,Kan:2021yoh}, which adds two new coordinates $(u,w)$ instead of the one $y$, and uses a Brinkmann supermetric
\be
g_{\mu\nu}\dd x^{\mu}\dd x^{\nu}-2U \dd u ^{2}+\dd u\, \dd w\,.
\ee
As shown in \cite{BenAchour:2022fif}, this improved lift identifies the new coordinate $u$ to the time coordinate $t$ along null geodesics, and its Killing vectors directly provide time-dependent conserved charges.

Finally, we would like to stress that the construction presented here strongly relies on the separation between configuration space and momenta. However, the classical description is of course invariant under any canonical transformation that could mix the two. In the case of a momentum-dependent potential, we could proceed as follows: find a canonical transformation that maps the Hamiltonian to a form like \eqref{curved_particle_model}, then apply the construction above and finally go back to the original variables.

\section{Two-dimensional field space geometries}
\label{sec:2d}

We now discuss the extension of the algebra \eqref{linear algebra} to charges which are quadratic in time. We will focus initially on the case of field space geometries which are two-dimensional and flat. Interestingly this is sufficient to cover many minisuperspace models of general relativity, as we will detail in section \ref{application}.

\subsection{Free case}

Let us consider the case of two-dimensional field spaces. Since every two-dimensional metric is conformally flat, we can use coordinates such that in Lorentzian signature (which is the field space signature of interest for us) we have
\be\label{conformal_2d}
\de\tilde{s}^2_\rm{mini} = -2 e^{-\varphi} \de \tilde u\, \de \tilde v\,.
\ee
Here all the dynamical information is encoded in the conformal factor $\varphi(\tilde{u},\tilde{v})$. The structure of the homothetic Killing vectors depends on the Ricci scalar curvature, which is given by
\be
R= -2e^{\varphi} \partial_{\tilde u} \partial_{\tilde v} \varphi\,.
\ee
In the non-flat case $R\neq0$ the number of independent homothetic Killing vectors varies on a case by case basis. For example we have the following cases:
\begin{itemize}	
\item For (A)dS$_2$, which corresponds to $\varphi = 2 \log(u\mp v)$, there are only the three Killing vectors
\be
\xi^\mu\partial_\mu =\left( \c_1 u^2 + \c_2 u +\c_3\right )\partial_u + \left (\pm \c_1 v^2 + \c_2 v \pm \c_3 \right )\partial_v \,,\q \pounds_\xi g_{\mu\nu} = 0\,,\q \c_i =\text{const.}\,,
\ee
where the upper sign is for AdS and the lower one is for dS.
\item For $\varphi = u^2 v$ there is only one homothetic Killing vector given by
\be
\xi^\mu\partial_\mu = -\lambda (u\partial_u - 2v \partial_v) \,,\q\q \pounds_\xi g_{\mu\nu} = \lambda g_{\mu\nu} \,.
\ee
\item For $\varphi = u v$, there is only one true Killing vector given by
\be
\xi^\mu\partial_\mu = - u \partial_u+  v\partial_v \,,\q\q \pounds_\xi g_{\mu\nu} = 0\,.
\ee
\end{itemize}  
The case of interest for us is that of a flat field space geometry.

In the free case (i.e. with vanishing potential), it is always possible to reduce the problem to that in a flat field space by a suitable choice of lapse. Indeed, in the Lagrangian \eqref{mini_lagrang} with coordinates such that the metric is \eqref{conformal_2d} we can simply choose the lapse to be $N=e^{-\varphi}$, so that $g_{\mu\nu}=\tilde{g}_{\mu\nu}/N$ is indeed flat. Alternatively, it might be the case that there are null coordinates such that $\tilde{g}_{\mu\nu}$ by itself is already flat. This is the case whenever we can separate $\varphi(\tilde{u},\tilde{v})=\varphi_1(\tilde{u}) + \varphi_2(\tilde{v})$. In this case we can simply take e.g. $N=1$, or any lapse of the form $N=N_1(\tilde{u})N_2(\tilde{v})$. Let us now assume that we are in the flat case, and use standard null coordinates to write
\be\label{Mink_de_s}
\de s^2_\rm{mini} = -2 \de u\, \de v\,.
\ee
This metric has three true Killing vectors (two translation and one boost), and one homothetic vector corresponding to the dilatation of the plane. For each of these vectors we can compute the factor $\lambda_{(i)}$ and the phase space function \eqref{phase space C}. In null coordinates we find
\bsub\label{HKV_mink}
\be
\text{time translation:}\quad&	& \xi_{\rm t}^\mu\partial_\mu =&\;\partial_u+\partial_v &\quad \lambda_{\rm{t}}=&\;0 
\quad&C_\rm{t} &=p_u + p_v\,,  \\
\text{space translation:}\quad&	& \xi_{\rm x}^\mu\partial_\mu =&\;\partial_u-\partial_v & \lambda_{\rm{x}}=&\;0
&C_\rm{x} &=p_u - p_v\,,  \\
\text{boost:}\quad&					& \xi_{\rm b}^\mu\partial_\mu =&\;u\partial_u-v\partial_v & \lambda_{\rm{b}}=&\;0
&C_\rm{b} &=u p_u -v  p_v\,,  \\
\text{dilatation:}\quad&				& \xi_{\rm d}^\mu \partial_\mu=&\;\f{1}{2}(u\partial_u+v\partial_v)	& \lambda_{\rm{d}}=&\;1
&C_\rm{d} &=\f{1}{2}(u p_u + v p_v)\,.\label{2d dilatation}
\ee
\esub 
Thanks to these Killing vectors, in addition to the Hamiltonian $\cQ_0\coloneqq H=-p_up_v$ we therefore get from \eqref{linear_charges} the four conserved charges $(\cQ_\rm{t}=C_\rm{t},\cQ_\rm{x}=C_\rm{x},\cQ_\rm{b}=C_\rm{b},\cQ_\rm{d}=C_\rm{d} - t \cQ_0)$ which are linear in time.

In two dimensions, the vanishing of the Ricci scalar and of the Riemann tensor are equivalent. This has an immediate implication on the existence of another set of charges. One can indeed show (see appendix \ref{HKV_prop}) that the following quantities have a vanishing third derivative:
\be
V_{(ij)} \coloneqq g_{\mu\nu} \xi_{(i)}^\mu \xi_{(j)}^\nu\,, \q\q \dddot{V}_{(ij)} =0\,.
\ee
The observables  $V_{(ij)}$ are similar to the 3d volume (or the scale factor) in FLRW cosmologies \cite{BenAchour:2019ufa,BenAchour:2020njq} or the metric components for the Kantowski--Sachs black hole ansatz  \cite{Geiller:2021jmg,Achour:2021dtj}. With the Killing vectors \eqref{HKV_mink} at hand we can explicitly calculate the scalar products $V_{(ij)}$ and find the non-trivial combinations
\be
V_{\rm{bt}}=-2 V_{\rm{dx}}=v-u\eqqcolon V_1\,,\q\!
V_{\rm{bx}}=-2 V_{\rm{dt}}=u+v\eqqcolon V_2\,,\q\!
V_{\rm{bb}}=-4 V_{\rm{dd}}=2 u v\eqqcolon V_3\,,\!
\ee
as well as the central elements $V_{\rm{xx}}=-V_{\rm{tt}}=2$. By construction the second derivative of the $V_{(ij)}$'s is a constant and in particular the following charges are conserved:
\be
\cQ_1\coloneqq V_1 + t\, C_\rm{x}\,,\q\q
\cQ_2\coloneqq V_2 + t\, C_\rm{t}\,, \q\q
\cQ_3\coloneqq V_3 + 4 t\, C_\rm{d} -2 t^2\cQ_0\,.
\ee
This can be checked using the fact that the brackets with the Hamiltonian are $\lb C_{(i)},\cQ_0\rb=\lambda_{(i)}\cQ_0$ and
\be\label{Vij Q0 algebra}
\lb V_1,\cQ_0\rb=-C_\rm{x}\,,\q\q\lb V_2,\cQ_0\rb=-C_\rm{t}\,,\q\q\lb V_3,\cQ_0\rb=-4C_\rm{d}\,.
\ee
Together with the charges $(\cQ_0,\cQ_\rm{t},\cQ_\rm{x},\cQ_\rm{b},\cQ_\rm{d})$ found previously, these charges form an 8-dimensional algebra isomorphic to the semi-direct sum $\big(\sl(2,\R) \oplus\R\big)\loplus\mathfrak{h}_2$, where $\mathfrak{h}_2$ is the two-dimensional Heisenberg algebra\footnote{Note that this algebra has four generators, which are analogous to the two positions and two momenta of a two-dimensional space obeying the Heisenberg commutation relations.}.

One should note that the charges $(\cQ_0,\cQ_1,\cQ_2,\cQ_3,\cQ_\rm{t},\cQ_\rm{x},\cQ_\rm{b},\cQ_\rm{d})$ are the initial conditions of the phase space functions $(\cQ_0,V_1,V_2,V_3,C_\rm{t},C_\rm{x},C_\rm{b},C_\rm{d})$, obtained by setting $t=0$. This explains why only the former set is conserved.

This procedure, which we have outlined here in the two-dimensional case, can be extended to a field space of arbitrary dimension, although it is of course not guaranteed to then lead to the same results. If the field space is non-flat and higher-dimensional we are indeed not guaranteed that homothetic Killing vectors exist, nor that the observables $C_{(i)}$ and $V_{(ij)}$, if they exist, form a closed algebra with the Hamiltonian. We give in appendix \ref{HKV_prop} the equation \eqref{condition on xi's} which has to be satisfied by the Killing vectors in order for this to be the case. Although this equation has a simple form, one cannot find an a priori condition on the field space geometry (and on the potential if this latter is non-trivial) which guarantees the existence of an algebra on phase space. With the construction presented here we have however the tools to investigate this in a case by case basis.

\subsubsection{Charge algebra}

In order to study the charge algebra it is convenient to slightly rewrite the above generators and first introduce the four quantities
\be\label{sl_2+R}
\cL_{-1} \coloneqq -\cQ_0\,,\q\q \cL_{0} \coloneqq \cQ_\rm{d}\,,&\q\q \cL_{1} \coloneqq \f{1}{2} \cQ_3\,, 
\q\q \cD_0 \coloneqq \cQ_\rm{b}\,,
\ee
which span $\sl(2,\R) \oplus\R$, as well as the four generators
\be\label{h_2}
\cS_{1/2}^\pm \coloneqq \f{1}{2}\big(\cQ_2 \pm \cQ_1\big)\,,&\q\q\cS_{-1/2}^\pm \coloneqq \f{1}{2}\big(\cQ_\rm{t} \pm \cQ_\rm{x}\big)\,,
\ee
which span $\mathfrak{h}_2$. These generators can finally be repackaged in the compact form
\bsub
\be
\cL_n&= \f{1}{4}\ddot f V_3 + \dot f C_{\rm d} - f \cQ_0\,,  &f(t) &= t^{n+1}\,, & n & \in \{-1,0,1\}\,,\\
\cS^\pm_s&= \f{1}{2}\Big(\big(\dot g V_2 + g C_\rm{t}\big) \pm\big(\dot g V_1 + g C_\rm{x}\big)\Big) \,,&g(t) &= t^{s+1/2}\,, & s & \in \left\{-\f{1}{2},\f{1}{2}\right\}\,.
\ee
\esub
We then find that they form the algebra
\bsub\label{charge_algebra}
\be
\lb \cL_n,\cL_m\rb &= (n-m) \cL_{n+m}\,,\label{sl2R part}\\
\lb \cL_n,\cD_0\rb&=0\,,\label{R part}\\
\lb \cS_s^\epsilon,\cS_{s'}^{\epsilon'}\rb &= 2 n \delta_{\epsilon+\epsilon'} \delta_{s+s'}\,,\label{h2 part}\\
\lb \cL_n,\cS_s^\pm\rb &= \left (\f{n}{2}-s\right ) \cS_{n+s}^\pm\,,\\
\lb \cD_0,\cS_s^\pm\rb &= \pm \cS_{s}^\pm\,,
\ee
\esub
with $n,m \in \{0,\pm1\}$ and $s,s' = \pm1/2$. As announced, this algebra is
\be\label{A algebra}
\cA=\big(\sl(2,\R) \oplus\R\big)\loplus\mathfrak{h}_2,
\ee
where the direct sum corresponds to the two brackets \eqref{sl2R part} and \eqref{R part}, the Heisenberg part comes from the centrally-extended bracket \eqref{h2 part}, and the remaining two brackets encode the semi-direct structure.

Because the physical system which we are describing has a four-dimensional phase space spanned by the coordinates $\{u,v,p_u,p_v\}$, these charges cannot be all independent. Indeed, one can show that they are related by the quadratic construction
\be\label{identities}
\cL_n = \sum_{k=-1/2}^{1/2} k \left (\f{n}{2}+k\right ) \big(\cS^+_k \cS^-_{n-k}+\cS^-_k \cS^+_{n-k}\big)\,,\q\q\cD_0= 2\sum_{k=-1/2}^{1/2} k\cS^-_k \cS^+_{-k}\,.
\ee
From this point of view, one may see the four generators $\cS^\pm_s$ of $\mathfrak{h}_2$ as the fundamental building blocks of the symmetry algebra. They represent the four constants of the motion which are preserved in the evolution, and any (possibly non-linear) combination of them will again be conserved on trajectories and generate a symmetry. The relations \eqref{identities} are analogous to the Sugawara construction.

In order to make contact with the symmetry algebra discussed in \cite{Geiller:2021jmg,Geiller:2020xze}, we can start from the linear charges $\cS_s^\epsilon$ to build two sets of generators as
\be
\cT^\pm_n\coloneqq \sum_{k=-1/2}^{1/2} k \left (\f{n}{2}+k\right ) \cS^\pm_k \cS^\pm_{n-k}\,.
\ee
Due to the Heisenberg central extension, these two abelian sets of generators do not commute and give
\be
\lb\cT^\pm_n, \cT^\pm_m\rb =0\,,\q\q\lb\cT^+_n, \cT^-_m\rb = \f{1}{2}(n-m)\cL_{n+m} -\f{1}{4}(n^2+m^2-nm-1) \cD_0\,.
\ee
Finally, these generators can be combined into
\be
\cP^\sigma_n\coloneqq\cT^-_n+\sigma\cT^+_n\,,
\ee
for which we find
\be
\lb \cP^\sigma_n , \cP^\sigma_m \big\}=  \sigma (n-m) \cL_{n+m}\,,\q\q
\lb \cP^\sigma_n , \cL_m \big\}=  (n-m) \cP^\sigma_{n+m}\,.
\ee
Together with the brackets \eqref{sl2R part}, these commutation relations show that $(\cP^\sigma_n,\cL_n)$ span an algebra $\mathfrak{g}_\sigma$ which, depending on the sign of $\sigma$, is either the isometry algebra $\mathfrak{g}_{\sigma>0}=\so(1,3)$ of dS, the isometry algebra $\mathfrak{g}_{\sigma<0}=\so(2,2)$ of AdS, or the isometry algebra $\mathfrak{g}_{\sigma=0}=\iso(2,1)$ of Minskowski. This latter, which is the (2+1)-dimensional Poincaré algebra, is the algebra which was studied in \cite{Geiller:2021jmg,Geiller:2020xze}.

\subsubsection{Symmetry transformations}

Let us now focus on the 8-dimensional algebra \eqref{A algebra} and investigate the corresponding symmetry group. Indeed we view the conserved charges as Noether charges, whose Poisson brackets have an exponentiated flow which generates finite symmetry transformations. Thus we start by computing their infinitesimal action given by their Poisson brackets with the field space coordinates $q^\mu=(u,v)$. The action on these null coordinates is found to be
\be
\begin{pmatrix}
\delta u\\
\delta v
\end{pmatrix}
=
\begin{pmatrix}
\df{1}{2}\dot f(t) u - f(t) \dot u + k u+ g(t)\vspace{0.3cm}\\
\df{1}{2}\dot f(t) v - f(t) \dot v - k v + h(t)
\end{pmatrix}\,,
\ee
where $k\in\R$, the function $f(t)$ is quadratic in $t$, and the functions $g(t)$ and $h(t)$ are linear in $t$. A straightforward calculation shows that these finite transformations are indeed symmetries of the action
\be\label{free action}
\cS= -\int \de t\,\dot u\dot v\,.
\ee
The algebra of these transformations is
\be
\big[\delta(f_1, g_1, h_1, k_1 ),\delta(f_2,g_2,h_2,k_2)\big]_{\rm{Lie}} = \delta(f_{12},g_{12},h_{12},k_{12})\,,
\ee
with
\bsub\label{Lie_commutator}
\be
f_{12}&= \dot f_1 f_2 -(1\leftrightarrow 2)\,,\\
k_{12}&=0\,,\\
g_{12}&= -\f{1}{2} g_1 \dot f_2 + f_2 \dot g_1 -k_1 g_1 -(1\leftrightarrow 2)\,,\\
h_{12}&= -\f{1}{2} h_1 \dot f_2 + f_2 \dot h_1 +k_1 h_1 -(1\leftrightarrow 2)\,.
\ee
\esub
This is equivalent to the charge algebra \eqref{charge_algebra}, up to the Heisenberg central extension. Each $\cL_n$ generates the coefficient $t^{n+1}$ in the function $f$, each $\cS^-_s$ generates the coefficient $t^{s+1/2}$ of $h$, and each $\cS^+_s$ generates the coefficient $t^{s+1/2}$ of $g$. Finally, $\cD_0$ corresponds to $k$.

These infinitesimal transformations can can be integrated to a finite group action. The exponentiated generators \eqref{sl_2+R} give the scaled M\"obius transformations, acting as the group $\SL (2, \R)\times \R $ on the null coordinates via
\be
\label{mobius_sym} 
\begin{pmatrix}
u\\
v\\
\end{pmatrix}\mapsto \big(\dot f\big)^{1/2}
\begin{pmatrix}
k\, u\\
k^{-1}v\\
\end{pmatrix} \circ f^{(-1)}=  \f{1}{c t + d}
\begin{pmatrix}
k\, u\\
k^{-1}v\\
\end{pmatrix}\circ f^{(-1)}\,,\q\q\text{where}\;\; f(t)=\f{a t + b}{c t + d}\,,
\ee
for constants coefficients $(k,a,b,c,d)$ such that $ad-bc=1$. The action of the Heisenberg generators \eqref{h_2} exponentiates to that of the abelian group $\R^4$ acting as
\be 
\label{abelian_symm}
\begin{pmatrix}
u \\
v \\
\end{pmatrix}\mapsto 
\begin{pmatrix}
u + \alpha t + \beta \\
v + \gamma t + \delta \\
\end{pmatrix}\,.
\ee
Again, one can check that these are symmetries of \eqref{free action}.

\subsection{Inclusion of a potential}

The algebra $\cA$ obtained above always exists in the free case. We now would like to tackle the case of a non-vanishing potential, and determine if this algebra (or a suitable deformation thereof) exists, or if a sub-algebra survives, or if there is no symmetry structure. In the previous section \ref{sec2:potential}, we have briefly discussed our strategy for the case with a non-vanishing potential, looking for Killing vectors along which the potential gets rescaled. With the explicit expression for the Killing vectors for the free theory at hand, we can now push this analysis further.

Considering an arbitrary minisuperspace model with two-dimensional field space, one can always pick null conformally-flat coordinates such that the metric on field space simply takes the form $\de\tilde{s}^2_\rm{mini} = -2 e^{-\varphi} \de \tilde u\, \de \tilde v$, so that the Lagrangian reads
\be\label{general null L with U}
\cL=-\f{1}{N}e^{-\varphi}\dot{\tilde{u}}\dot{\tilde{v}}-N\tilde{U}\,,
\ee
where the potential $\tilde{U}$ and the conformal factor $\varphi$ follow from the definition of the minisuperspace, while the lapse $N$ can be chosen arbitrarily. We see two natural possibilities for studying this Lagrangian.
\begin{itemize}
\item[(1)] Choose $N=1/\tilde{U}$ so as to trivialize the potential term.
\begin{itemize}
\item[(1.1)] If the resulting conformal factor $\tilde{U} e^{-\varphi}$ is such that the supermetric is flat then the whole algebra $\cA$ survives by the above construction.
\item[(1.2)] If the resulting supermetric is not flat, one should go back to the beginning of the analysis to study the conformal Killing vectors $\xi_{(i)}$, construct the objects $C_{(i)}$ and $V_{(ij)}$, and compute their algebra.
\end{itemize}
\item[(2)] Choose $N=e^{-\varphi}N_1(\tilde{u})N_2(\tilde{v})$ such that the supermetric is flat, and then study whether the resulting potential $U=N\tilde{U}$ satisfies $\pounds_\xi U=-\lambda U$ for \eqref{HKV_mink}.
\end{itemize}
This procedure has to be implemented on a case by case basis, and we give the examples of physical interest in section \ref{application} below. In general, we expect that demanding that the homothetic Killing vectors of the metric also be Killing vectors of the potential will reduce the number of admissible vectors, and therefore also reduce the size of the symmetry algebra.

In order to illustrate this, let us consider an example, relevant for the Bianchi models, where the shifted potential $U(q)=N(q)\tilde{U}(q)-U_{0}$, as defined earlier in \eqref{new potential}, is after a suitable choice of lapse a monomial in the two null coordinates, i.e. of the form $U=u^nv^m$. This leads to three possibilities:
\begin{itemize}
\item[(1)] If $n=0$, $m\neq 0$, then there are two vectors satisfying \eqref{potential_condition}, given by
\be
\xi_1 = \xi_\rm{t}+\xi_\rm{x}\,,\q\quad \xi_2 = (2+m)\xi_\rm{b}+2m\xi_\rm{d}\,,\q\quad\lambda_1=0\,,\q\quad\lambda_2=2m\,.
\ee
\item[(2)] If $n\neq0$, $m=0$, we have
\be
\xi_1 = \xi_\rm{t}-\xi_\rm{x}\,,\q\quad \xi_2 = (2+n)\xi_\rm{b}-2n\xi_\rm{d}\,,\q\quad\lambda_1=0\,,\q\quad\lambda_2=-2n\,.
\ee
\item[(3)] If $n,m\neq 0$, then there is a single admissible vector given by
\be\label{case 3 for U}
\xi_1 =(2+n+m)\xi_\rm{b}+2(m-n)\xi_\rm{d}\,,\q\q\lambda_1=2(m-n)\,.
\ee
\end{itemize}
This shows that the existence and the number of the linear charges $\cQ_{(i)}$ defined in \eqref{linear_charges} depends on the potential. Then, in order for the quadratic charges to exist and close the algebra, we need to have closed brackets as in \eqref{Vij Q0 algebra} between the phase space functions $V_{(ij)} = g_{\mu\nu} \xi_{(i)}^\mu \xi_{(j)}^\nu$, the Hamiltonian $\cQ_0$, and the functions $C_{(i)}$. For example in the third case above, with $n,m\neq0$ we compute
\bsub
\be
V_{11}=\xi^2_1 &= (2+n+m)^2 \xi_\rm{b}^2 -4(m-n)^2 \xi_\rm{d}^2 = 4(1+m)(1+n) V_3\,,\\
C_1=p_\mu \xi^\mu_1 &= (2+n+m) C_\rm{b} -2(m-n) C_\rm{d}\,, \\
\{C_1,V_{11}\} &= 2(n-m)V_{11}\,,\\
\{\cQ_0,V_{11}\} &= 8(1+m)(1+n) C_\rm{d}\,.
\ee
\esub
In order for the last bracket to close on $C_1$ the potential must be such that $2+n+m\stackrel{!}{=}0$. In this case the algebra of the charges reduces to $\sl(2,\R)\subset\cA$. This situation, which corresponds to a potential $U\propto u^nv^m$ with $(n+m)=-2$ which is homogeneous of degree $-2$, is the two-dimensional equivalent of the conformal potential $U=1/q^2$. We shall remark that also for the cases $m=-1$, $n=-1$ the algebra closes, but $V_{11}$ is trivial and the algebra is two dimensional, containing only $\cQ_0$ and $C_1$. 

We note in passing that the case of a model with one-dimensional field space is trivial and always admits an $\sl(2,\R)$ algebra. Indeed, if the Lagrangian is $\cL=\tilde{g}(q)\dot{q}^2/N-N\tilde{U}$ one can choose $N=1/\tilde{U}$ and work with the rescaled metric $g=\tilde{g}\,\tilde{U}$. The resulting phase space with Hamiltonian $H=g^{-1}p^2/2+1$ then supports an $\sl(2,\R)$ algebra spanned by $(C,V,\cQ_0)$ with
\be\label{1d case}
\cQ_0=H-1\,,\q\q V=\xi^2\,,\q\q C=g^{-1}p\,\xi\,,\q\q\xi=\c_1+\f{\c_2}{2}\int^q\de q'\sqrt{g(q')}\,,
\ee
where $\c_1$ and $\c_2\neq0$ are otherwise arbitrary constants. This simple setup is here deliberately presented in an analogous way to the case of a higher-dimensional field space.

\subsection{Applications}
\label{application}

Let us now apply the construction presented above to some examples of supermetrics on field space inherited from minisuperspace models of general relativity. We study Kantowski--Sachs (KS) cosmologies describing the black hole interior, Bianchi models, and then FLRW cosmology with a scalar field.

\subsubsection{Black holes and Bianchi models}

The Bianchi classification is based on the nature of the three-dimensional Lie algebra of spacetime vector fields which leaves the spatial triads $e^i_\mu(x)$ invariant. This classification does however leave the internal metric $\gamma_{ij}$ completely arbitrary. Here we are interested in the choices which solve the vector constraint, as explained in section \ref{sec:homogeneous}. The classification restricted to the case of metrics solving the vector constraint is given in appendix \ref{Bianchi_metrics}. For instance KS cosmologies do not belong to the family of Bianchi models, because their spatial slices do not admit a three-dimensional algebra of isometries. They are nonetheless minisuperspaces satisfying the vector constraint, with a field space geometry actually analogous to that of Bianchi models, and they therefore naturally enter the scope of the present analysis.

All the Bianchi models have a two-dimensional field space, appart from Bianchi I and II which we discuss in section \ref{sec:Bianchi I II}. The KS model and the Bianchi models\footnote{To the list which is studied here we should add the model VI$_{h}$. This has however the same symmetry properties as VI$_0$, but its treatment involves some lengthy expressions. The interested reader can look at \cite{francesco-thesis}, where the formulas pertaining to the VI$_h$ model are reported.} with two-dimensional field space have an internal metric given by
\bsub
\be
&\text{(III, VI$_0$, VIII, IX, KS)} &&\to& \gamma_{ij} &= \rm{diag}\left(a^2,b^2,b^2\right)\,,\\
&\text{(IV, V, VII)} &&\to & \gamma_{ij} &= \rm{diag}\left (a^2,\f{a^4}{b^2},b^2\right)\,.
\ee
\esub
The kinetic term of the symmetry-reduced action \eqref{Einstein_mini}, from which we read the supermetric $\tilde{g}_{\mu\nu}$, depends only on the internal metric $\gamma_{ij}$. We find
\be
\f{\cV_0}{4GN}\sqrt{\gamma}\Big((\gamma^{ij} \dot \gamma_{ij})^2+ \dot \gamma_{ij} \dot \gamma^{ij}\Big)= 
- \df{2 \cV_0}{GN}\,\tilde{g}_{\mu\nu}\dot{q}^\mu\dot{q}^\nu\,,
\ee
with the configuration variables $q^\mu=(a,b)$ and the following supermetrics:
\bsub
\be
&\text{(III, VI$_0$, VIII, IX, KS)} &&\to& \tilde{g}^{\mu\nu} &= \begin{pmatrix}0&b\\b&a\end{pmatrix}\,,\\
&\text{(IV, V, VII)} &&\to& \tilde{g}^{\mu\nu} &=\f{a}{b^2}\begin{pmatrix}2b^2&ab\\ab&-a^2\end{pmatrix} \,.
\ee
\esub
Now, we can use the following change of coordinates to the null conformal parametrization \eqref{conformal_2d}:
\bsub\label{u v tilde map}
\be
&\text{(III, VI$_0$, VIII, IX, KS)} &&\to& & \tilde u = \f{2\sigma\sqrt{2}}{\sqrt{3}}\, a \sqrt{b}\,, && \tilde v = \df{2 \sqrt{2}}{\sigma\sqrt{3}}\, b^{3/2}\,, \\
&\text{(IV, V, VII)} &&\to& & \tilde u= \f{2\sigma\sqrt{2}}{\sqrt{3}}\, a^{(3+\sqrt 3)/2} b^{-\sqrt{3}/2}\,, \quad&& \tilde v =\df{2 \sqrt{2}}{\sigma\sqrt{3}}\, a^{(3-\sqrt 3)/2} b^{\sqrt{3}/2}\,,
\ee
\esub
where $\sigma$ is an arbitrary real parameter. With this, we find that all the models under consideration, i.e. (III, IV, V, VI$_0$, VII, VIII, IX, KS), lead to actions which can all be written in the same compact form
\be\label{all Bianchi actions}
\cS=\int\de t\left[-\f{\cV_0}{GN}\dot {\tilde u} \dot {\tilde v}-N\tilde{U}\right],
\ee
and which are distinguished by their potential $\tilde{U}$. Comparing to the general form of the Lagrangian \eqref{general null L with U}, we see that all these models also have a vanishing conformal factor, $\varphi=0$. Finally, one should note that here we are keeping track of the dimensional constants $\cV_0$ and $G$.

We now distinguish two families of models depending on their potential:
\begin{itemize}
\item[(1)] The models (III, V, VI, KS), for which the potential separates in the product of two functions of the null directions.
\item[(2)] The models (IV, VII, VIII, IX), for which this does not happen.
\end{itemize}

\paragraph{Models (III, V, VI, KS).}

In this case we simply need to take $N=1/\tilde{U}$. Indeed, this trivializes the potential term in \eqref{all Bianchi actions}, and leads to a flat field space metric for which we already know the construction of the algebra $\cA$ following the previous sections. In order to import verbatim these previous results, we may need to change coordinates one last time in order to put the (already) flat supermetric in the precise form \eqref{Mink_de_s}.

Let us illustrate this construction by focusing on the Bianchi III model. The other models are all explicitly analyzed in appendix \ref{Bianchi_metrics}. In conformal null coordinates, the action \eqref{action example} for the Bianchi III model reads
\be
\cS_\text{III}=\int\de t\left[-\f{\cV_0}{GN}\dot {\tilde u} \dot {\tilde v}-N\tilde{U}\right]\,,\q\q \tilde{U} = - \f{3^{1/3} \cV_0}{\sigma^{4/3}GL_s^2} \df{\tilde u}{\tilde{v}^{1/3}}\,.
\ee
Taking $N=\cV_0^2/(G^2L_s^2\tilde{U})$ (where the dimensional factors enter the relation between the Hamiltonian $H$ and the charge $\cQ_0$) and changing coordinates to
\be
\label{conf_to_null_exmpl}
u =- \df{3\tilde u^2}{4\sigma^2}    =-2 a^2 b\,,\q\q
v = \big(3\sigma^2\tilde v^2\big)^{1/3}  =2 b \,,
\ee
finally puts the action in the form
\be
\cS_\text{III}=\int\de t\left[-\dot{u}\dot{v}-\f{\cV_0^2}{G^2L_s^2}\right]\,.
\ee
This is precisely the setting discussed in section \ref{sec:2d} for the free case, but now with a constant potential $U_0$ which simply shifts the Hamiltonian. Using the canonical transformation \eqref{conf_to_null_exmpl} we can import the charges obtained previously to find the explicit expression for the conserved charges of the Bianchi III model. These are
\bsub
\be
C_\rm{t}&=-\f{p_a}{4 a b} +   \f{a p_a+ 2 b p_b}{4 b}\,,&
C_\rm{b}&= -a p_a - b p_b\,, \\
C_\rm{x}&=-\f{p_a}{4 a b} -   \f{a p_a+ 2 b p_b}{4 b}\,,&
C_\rm{d}&= \f{b p_b}{2}\,,\\
V_1&= 2b\left (1-a^2 \right )\,,&
V_3&= 8 a^2 b^2\,, \\
V_2&= 2b\left (1+a^2 \right )\,,&
\cQ_0&=H- \f{\cV_0^2}{G^2L_s^2} = \f{p_a (a p_a-2 b p_b)}{16 a b^2}\,.
\ee
\esub
By construction these charges obey the algebra $\cA$, which is therefore fully implemented in the models (III, V, VI, KS)\footnote{In \cite{Geiller:2021jmg,Geiller:2020xze} we have used for the KS model the line element $\de s_\rm{KS}^2 = -N^2 \de t +(8v_2/v_1)\de x^2 + v_1 L_s^2 \de \Omega^2$, related to the one used here by $(v_1,8v_2/v_1)=(b^2,a^2)$ and $(v_1,v_2)=(u^2/4,uv)$. On these variables we found that the transformation generated by $\cL_n$ is $\delta_fv_i=\dot{f}v_i-f\dot{v}_i$ with $f(t)= t^{n+1}$ (these are the superrotations when $n\in\Z$), while $\cT^-_n$ acts as $\delta_g(v_1,v_2)=(0,\dot{g}v_1-g\dot{v}_1)$ with $g(t)= t^{n+1}$ (these are the supertranslations when $n\in\Z$). The action of $\cT^+_n$, which we have not studied in these two references, is $\delta_g(v_1,v_2)=(1,v_2/(2v_1))(v_2\dot{g}+g(\dot{v}_1v_2/v_1-2\dot{v}_2))/4$ with $g(t)= t^{n+1}$. This is also a symmetry of the symmetry-reduced KS action.}
. Note that we have kept track of the dimensional factors in order to have a shift of the Hamiltonian defining $\cQ_0$ with dimension of inverse length squared.

\paragraph{Models (IV, VII, VIII, IX).} The study of this family is a little more subtle. First, one should notice that setting $N=1/\tilde{U}$ for these models does not lead to a flat supermetric. Second, it turns out that this non-flat supermetric does not admit any Killing vectors, so the method fails when the lapse is chosen as the inverse of the full potential. However, it turns out that for these models the potential is actually a sum of two or three monomials of the form $\tilde u^n \tilde v^m$. This opens the possibility of choosing the lapse so as to set one of these terms of the potential to a constant. By construction, for these models since $\varphi=0$ such a choice of lapse will lead to a flat field space metric. One can then choose coordinates to put this metric in the form \eqref{Mink_de_s}, and finally study the condition \eqref{potential_condition} for the remaining terms in the potential.

Let us illustrate this construction on the general action of the form
\be
\cS= \f{\cV_0}{G} \int \de t \left [-\f{1}{N}\dot {\tilde u} \dot {\tilde v}+ N\Big (\c_1 \tilde u^{n_1} \tilde v^{m_1} + \c_2 \tilde u^{n_2} \tilde v^{m_2}\Big )\right ]\,.
\ee
Without loss of generality we can now choose the lapse and the coordinates (with $m_i, n_i \neq -1$) as
\be
N = \f{\cV_0}{G} \f{1}{\tilde u^{n_1} \tilde v^{m_1}}\,,\q\q
u = \df{\tilde u^{n_1+1}}{1+n_1} \,,\q\q
v = \df{\tilde v^{m_1+1}}{1+m_1} \,,
\ee
to obtain
\be
\cS=\int \de t \Big [-\dot {u} \dot {v}-(U+U_0)\Big ]\,,
\ee
with the potential
\be
U_0=-\f{\cV_0^2}{G^2}\c_1\,,\q\q U=-\f{\cV_0^2}{G^2}\c_2 \big((n_1+1)u\big)^{\f{n_2-n_1}{n_1+1}} \big((m_1+1)v\big)^{\f{m_2-m_1}{m_1+1}}\,.
\ee
We are now exactly in the case \eqref{case 3 for U} discussed above. There is a single homothetic Killing vector, and an $\sl(2,\R)$ subalgebra of $\cA$ survives, spanned by the charges $(C_1,V_{11},\cQ_0)$, if and only if
\be\label{power_condition}
2+ n_1 +n_2+ m_1 +m_2 + n_1 m_2 + n_2 m_1 =0\,.
\ee
It turns out that this equation is satisfied for the potentials of the Bianchi models VIII and IX, but not for the Bianchi models IV and VII.

In summary, we have shown that the full algebra $\cA$ \eqref{A algebra} can be obtained for the models (III, V, VI, KS), while for the models (VIII, IX) only the conformal subalgebra survives. The Bianchi models IV and VII are particular in the sense that they cannot be assigned an $\sl(2,\R)$ algebra of phase space functions with the construction presented here. This deserves more investigation. A classical change of variable does not alter this conclusion, as our analysis is completely covariant on the field space, however a different choice of internal metric (e.g. non diagonal) might change the results. 

\subsubsection{FLRW cosmology with scalar field}

In the case of FLRW cosmology with curvature $k$, we only have a single scale factor $a$, so the field space is one-dimensional and \eqref{1d case} easily applies. This also extends to the case of a non-vanishing cosmological constant, where it reproduces the results of \cite{Achour:2021lqq}. On the other hand, without any matter content the dynamics itself is trivial.

Let us therefore consider the addition of a scalar field $\Phi$, and apply the above formalism to the two-dimensional field space parametrized by the coordinates $\Phi$ and $a$. With the line element
\be
\de s^2_{\rm{FLRW}} = -N^2 \de t^2 + a(t)^2 \left (\f{\de r^2}{1-k r^2} + r^2 (\de \theta^2+\sin^2 \theta\, \de \phi^2)\right )\,,
\ee
the symmetry-reduced action reads
\be
\cS_{\rm{FLRW}}= \cV_0 \int \de t\left [\f{1}{2N}\left(a^3 \dot \Phi^2 -\f{3 a \dot a^2}{4\pi G}\right) + N\f{3 k a}{8 \pi G}  \right ]\,, \q\q \cV_0 = \int_\Sigma \de^3x\f{r^2 \sin \theta}{\sqrt{1-k r^2}}\,.
\ee
This reduced action is of the same form as the others treated above in this article, the supermetric and the potential beeing
\be
\de s^2_\rm{mini} = \f{\cV_0 a^3 }{N} \de \Phi^2 -\f{3 \cV_0 a}{4 \pi G N} \de a^2\,, \q\q \tilde{U} = -\f{3  k a\cV_0 }{8 \pi G}\,.
\ee
Along the lines of the above discussion in the case of the Bianchi models, we choose the lapse which sets the potential to a constant $N\tilde{U}=U_0$, i.e.
\be
N = -\f{8 \pi \cV_0}{3 G a }\,,\q\q U_0  = -k \f{\cV_0^2}{G^2}\,,
\ee
and find a map to the null coordinates \eqref{Mink_de_s} given by
\be
\label{FLRW_conf}
\begin{array}{rlcrl}
 u \!\!\!&= \df{3}{16 \pi}  a^2 e^{2 \kappa \Phi/3}\,,&\q\q&p_u\!\!\!&=\df{4\pi}{3\kappa}\df{ap_a\kappa+3p_\Phi}{a^2}e^{ -2 \kappa \Phi/3}\,,\\[9pt]
 v \!\!\!&= \df{3}{16 \pi} a^2 e^{ -2 \kappa \Phi/3}\,,&\q\q&p_v\!\!\!&=\df{4\pi}{3\kappa}\df{ap_a\kappa-3p_\Phi}{a^2}e^{ 2 \kappa \Phi/3}\,, \\
\end{array}
\ee
with $\kappa^2 = 12 \pi G$. The charges and their algebra are found by mapping $\cA$ under this transformation.

\section{Three-dimensional field space geometries}
\label{sec:Bianchi I II}

We now briefly discuss an application of the framework to three-dimensional field spaces, which appear e.g. in the Bianchi I and II models. Using the line element \eqref{bianchi II}, we find that the reduced action for these Bianchi models takes the form
\be
\cS=-\f{\cV_0}{G}\int\de t\left[\f{2}{N}\big(a\dot{b}\dot{c}+b\dot{a}\dot{c}+c\dot{a}\dot{b}\big)+\eps\f{N}{2}\f{a^2}{bc}\right]\,,
\ee
where Bianchi I has a vanishing potential with $\eps=0$, and Bianchi II has $\eps=1$. The three-dimensional field space metric is diagonalized by the change of (field space) coordinates
\be
a=e^{\sqrt{2}\,u-2w}\,,\q\q b=e^{v+w}\,,\q\q c=e^{\sqrt{2}\,u-2v}\,,
\ee
leading to
\be
\cS
&=-\f{\cV_0}{G}\int\de t\left[\f{4}{N}e^{2\sqrt{2}\,u-v-w}\big(\dot{u}^2-\dot{v}^2-\dot{w}^2\big)+\eps\f{N}{2}e^{2\sqrt{2}\,u+v-7w}\right]\cr
&=-\f{\cV_0}{G}\int\de t\left[\f{1}{2N}\tilde{g}_{\mu\nu}\dot{q}^\mu\dot{q}^\nu+\eps\f{N}{2}e^{2\sqrt{2}\,u+v-7w}\right]\,,
\ee
where we denote the new coordinates by $q^\mu=(u,v,w)$ even though these are not null coordinates as in the two-dimensional case. The field space metric $\tilde{g}_{\mu\nu}$ is not flat and, in the case of Bianchi II, we have moreover to take into account the presence of the potential.

\subsection{Bianchi I}

In the case of the Bianchi I model, we have a free theory with $\eps=0$. We identify two interesting choices of lapse, which both lead to an $\sl(2,\R)$ symmetry algebra.

A first possibility is to choose $N=e^{2\sqrt{2}\,u-v-w}$, so that the conformally-rescaled metric $g_{\mu\nu}$ is that of $(2+1)$-dimensional Minkowski space. The metric then admits six exact Killing vectors corresponding to the six Poincaré isometries. They are given by
\be
\xi=t^\mu\partial_\mu+{\eps^\mu}_{\nu\rho}\partial_\mu q^\nu\ell^\rho\,,
\ee
where $t=(t_1,t_2,t_3)$ are translations and $\ell=(\ell_1,\ell_2,\ell_3)$ Lorentz transformations. The metric also admits an homothetic Killing vector, corresponding to a dilatation with
\be
\xi_\rm{d}=u\partial_u+v\partial_v+w\partial_w\,,\q\q\lambda_{\rm{d}}=2\,,\q\q C_\rm{d}=u p_u + v p_v+ w p_w\,,
\ee
which generalizes the scaling symmetry \eqref{2d dilatation} of the two-dimensional models to the three-dimensional case. Starting with these 7 conformal Killing vectors, we follow the same procedure as above and define the quantities $C_{(i)}=p_\mu\xi^\mu_{(i)}$ and $V_{(ij)}=g_{\mu\nu}\xi^\mu_{(i)}\xi^\nu_{(j)}$. There are 7 non-trivial $C_{(i)}$'s and 10 non-trivial $V_{(ij)}$'s given by
\be
\begin{pmatrix}
p_u\\
p_v\\
p_w\\
vp_w-wp_v\\
up_w+wp_u\\
-up_v-vp_u\\
u p_u + v p_v+ w p_w
\end{pmatrix}\in C_{(i)}\,,\q\q
\begin{pmatrix}
u\\
v\\
w\\
uv\\
uw\\
vw\\
u^2-v^2\\
u^2-w^2\\
v^2+w^2\\
u^2-v^2-w^2
\end{pmatrix}
\in V_{(ij)}\,. \label{BianchiI_alg}
\ee
Computing the triply-iterated bracket of these $V_{(ij)}$'s with $H$ reveals that $\dddot{V}_{ij}=0$, in agreement with equation \eqref{triple dot V Riemann} and the fact that the field space geometry is flat. In order to have phase space functions which form a closed algebra with $C_{(i)}$, $V_{(ij)}$ and $H$, we need however to include six new functions produced by the time evolution $\dot{V}_{(ij)}$, as well as another six new functions produced by $\ddot{V}_{(ij)}$. These are given explicitly by
\be
\begin{pmatrix}
up_u+vp_v\\
up_u+wp_w\\
vp_v+wp_w\\
up_v-vp_u\\
up_w-wp_u\\
vp_w+wp_v\\
\end{pmatrix}
\in\dot{V}_{(ij)}\,,\q\q
\begin{pmatrix}
p_up_v\\
p_up_w\\
p_vp_w\\
p_u^2-p_v^2\\
p_u^2-p_w^2\\
p_v^2+p_w^2\\
\end{pmatrix}
\in\ddot{V}_{(ij)}\,.
\ee
At the end of the day, we find $7+10+1+6+6=30$ generators forming a closed algebra. There are in particular three generators forming an  $\sl(2,\R)$ algebra. This symmetry algebra is spanned by
\be
C_\rm{d}=up_u+vp_v+wp_w\,,\q\q V_\rm{dd}=u^2-v^2-w^2\,,\q\q H=p_u^2-p_v^2-p_w^2\,,
\ee
which is in agreement with the results of \cite{BenAchour:2019ywl} concerning the existence of an $\sl(2,\R)$ algebra in the case of the Bianchi I model. We should note that the 30 phase space functions given above are not all independent: they satisfy linear as well as quadratic dependency relations. For instance, one recovers $V_\rm{dd}$ as a linear combination of three other elements of $V_{(ij)}$. Similarly, one recovers $C_\rm{d}$ as a linear combination of three other elements of $\dot{V}_{(ij)}$, and finally one recovers $H$ as a linear combination of three other elements of $\ddot{V}_{(ij)}$. These linear and quadratic dependency relations are obvious from the fact that $C_{(i)}$ and $V_{(ij)}$ contain the phase space variables $(u,v,w,p_u,p_v,p_w)$, which can therefore be combined in order to reproduce all the other 24 generators. The fact that the above 30 functions have closed brackets is nonetheless a non-trivial statement.

Alternatively, one could choose the lapse $N=1$, and work with the non-flat field space metric $g_{\mu\nu}=e^{2\sqrt{2}\,u-v-w}\text{diag}(1,-1,-1)$. This metric admits three homothetic vectors given by
\be
\xi=\f{\c_1}{2\sqrt{2}}\partial_u+\c_2\partial_v+\c_3\partial_w\,,\q\q\lambda=\c_1-\c_2-\c_3\,,
\ee
where the $\c_i$'s are constants (which can evidently be chosen so as to obtain true Killing vectors). Out of the six possible observables $V_{(ij)}$ obtained by contracting these homothetic vectors, we find that only three are non-trivial, and furthermore all proportional to $V=e^{2\sqrt{2}\,u-v-w}=abc$, which is simply the 3d volume of the spatial slices. This function forms an $\sl(2,\R)$ algebra together with $H$ and the function $C=2\sqrt{2}\,p_u+p_v+p_w$, which is the sum of the three $C_{(i)}$'s and defines the isotropic dilatation generator. As expected, this is simply a canonical transformation of the algebraic structure found with the previous choice of lapse. This $\sl(2,\R)$ algebra controls the evolution of the isotropic volume $V=abc$, whose speed is given by the dilatation generator $C$ and acceleration by the Hamiltonian $H$.

\subsection{Bianchi II}

Finally, we turn to the subtler case of the Bianchi II model, with both a non-trivial potential and an a priori non-flat field space geometry. In order to treat this case, we choose $N=2/e^{2\sqrt{2}\,u+v-7w}$ so as to set the potential to a constant. Note that here, for simplicity's sake, we do not include the dimensional factors in the lapse, so the value of the Hamiltonian will be shifted simply by $1$. The conformally-rescaled metric then admits one true and two homothetic Killing vectors with
\bsub\label{HKV_II}
\be
\xi_1=&\;\partial_v &\quad \lambda_1=&\;0 \quad&C_1 &=p_v\,,  \\
\xi_2=&\;\partial_u & \lambda_2=&\;4\sqrt{2}&C_2 &=p_u\,,  \\
\xi_3=&\;\partial_w & \lambda_3=&\;-8 &C_3 &=p_w \,.
\ee
\esub
All three $V_{(i)}$'s are proportional to $V=e^{4\sqrt{2}\,u-8w}=a^4$. Together with the shifted Hamiltonian $\cQ_0=H-1$, this leads to the charge algebra
\bsub
\bg
\lb C_{(i)},C_{(j)}\rb=0\,\q\q\lb C_1,\cQ_0\rb=0=\lb C_1,V\rb\,,\q\q\lb V,\cQ_0\rb=8\big(\sqrt{2}C_2+2C_3\big)\,,\\
\lb C_2,\cQ_0\rb=4\sqrt{2}\cQ_0\,,\q\lb C_3,\cQ_0\rb=-8\cQ_0\,,\q\lb C_2,V\rb=-4\sqrt{2}V\,,\q\lb C_3,V\rb=8V\,,
\eg
\esub
which therefore also contains an $\sl(2,\R)$ spanned by $C\coloneqq\sqrt{2}C_2+2C_3$, the volume $V$, and $\cQ_0$. This algebra encodes the evolution of the scale factor $a$.

\section{Perspectives}

In this paper we have studied the symmetry structure of gravitational minisuperspaces. For this, we have considered reductions of general relativity to one-dimensional models, for which the variable components of a given space-time metric ansatz depend on a single coordinate, chosen as the time coordinate. In that case, the Einstein--Hilbert action reduces to mechanical systems. Imposing that the metric ansatz satisfies the ADM vector constraint ensures that the equations of motion resulting from the reduced action are equivalent to the full Einstein equations for the original line element.

We have based our study of the symmetries on the observation that the dynamics is described by a Lagrangian of the form
\be\tag{\ref{mini_lagrang}}
\cL =\f{1}{2} g_{\mu\nu} \dot{q}^\mu \dot{q}^\nu - U (q)\,.
\ee
Here the variables $q_{\mu}$ are components of the space-time line element, i.e. the dynamical fields of the minisuperspace models. They depend on the time coordinate $t$ and this Lagrangian describes their evolution in time. It depends on a metric in field space, or  \textit{supermetric}, $g_{\mu\nu}$. The isometries of this supermetric and the scaling properties of the potential $U$ are directly connected to the symmetries of the mechanical model. In particular, once we identify a homothetic vector field $\xi$ satisfying the conditions
\be
\pounds_\xi g_{\mu\nu}   =\lambda  g_{\mu\nu} \,,  
\q\q
\pounds_\xi  U   =-\lambda U\,,
\ee
then the quantity $\cQ=\xi^\mu p_\mu - t \lambda H$ is conserved and, in turn, it generates a symmetry of the system. However, the non-covariance of the conformal properties of the supermetric and the potential under (field-dependent) changes of lapse render a systematic analysis quite involved. This subtlety reflects the deeper fact that the conserved charges associated with these minisuperspace symmetries are local only with respect some particular choices of clocks, while they might depend on the history of the system through a non-local factor for a more general time coordinate \cite{francesco-thesis}.

The majority of cases of interest (the black hole interior and some Bianchi models) fall in the class of two-dimensional field spaces, where we can always pick coordinates in which the supermetric is conformally flat. The advantage of this observation is twofold. On the one hand, we can  easily handle the question of the existence of homothetic killing vectors, and the related choice of lapse, and on the other hand, it allows us to determine the presence of charges that are quadratic in time. These are built for flat supermetrics starting from the scalar products between vector fields, $V_{(ij)} = \xi_{(i)\mu} \xi^\mu_{(j)}$. Thanks to the study of two-dimensional superspaces, we were able to identify the existence of an 8-dimensional algebra $\cA=\big(\sl(2,\R) \oplus\R\big)\loplus\mathfrak{h}_2$ for FLRW cosmology with a scalar field, Kantowski--Sachs cosmologies and the Bianchi models III, V and VI. Therein the $\sl(2,\R)$ sector is a generalization of the CVH algebra originally found in flat cosmologies. For the Bianchi VIII and XI models, only the CVH sector survives, while for the models IV and VII we find that the non-trivial potential spoils the construction and that our procedure does not produce any algebra. The reason behind this peculiar property of the Bianchi models IV and VII is still unclear.

We have also included a brief study of the Bianchi I and II models, which represent two examples of three-dimensional superspaces. For the Bianchi I model, we find an algebra with 30 generators \eqref{BianchiI_alg}, while Bianchi II leads to a four-dimensional algebra, still containing an $\sl(2,\R)$ subalgebra.

Despite their apparent simplicity, the minisuperspaces contain a rich symmetry structure. However, a clear understanding of the origin and physical role of this structure is still needed. For black holes and cosmology, a subtle relationship with the spatial boundaries and the scaling properties of the model have already been unravelled and studied \cite{francesco-thesis, Geiller:2020xze, BenAchour:2019ufa}, and a possible consequence on perturbation theory has been recently pointed out \cite{BenAchour:2022uqo}. It would be interesting to see to what extent this feature generalizes to the other Bianchi models.

For flat cosmologies and black holes, the infinite-dimensional extension of the symmetry group gives a solution generating tool, that allows to turn on an effective cosmological constant or a scalar field \cite{Achour:2021dtj, BenAchour:2020xif, Geiller:2021jmg}. The infinite-dimensional group acts as a rescaling of the coupling constants of the theory, in a manner reminiscent of a renormalization group flow. A generalization to the Bianchi models is an ongoing work and might help us obtain a better understanding of these symmetry structures.

Finally, the simplicity of the mechanical models allows for a straightforward quantization, where the conserved charges naturally provide a (full) set of Dirac observables \cite{francesco-thesis,Sartini:2021ktb}, that can be straightforwardly quantized by means of the representation theory of the symmetry group. The requirement of protection of the symmetry gives also a valuable tool to discriminate between different effective dynamics \cite{BenAchour:2018jwq,BenAchour:2019ywl,Geiller:2020xze,Sartini:2021ktb}. It also opens the door towards a possible emerging description of spacetime itself, out of the symmetry group \cite{Aldaya:1999yn}.

In the end, a deeper comprehension of the origin of the minisuperspace symmetries could enlighten many aspects of the holographic properties of general relativity, and hopefully help understand the proper notion of observables for a quantum theory of gravity.

\newpage

\section*{Appendices}
\addcontentsline{toc}{section}{Appendices}
\renewcommand{\thesubsection}{\Alph{subsection}}
\setcounter{section}{1}
\counterwithin*{equation}{subsection}
\renewcommand{\theequation}{\Alph{subsection}.\arabic{equation}}

\subsection{ADM approach to the CVH algebra}
\label{ADM}

We have explained in the core of the article how homogeneous minisuperspace models lead to mechanical Lagrangians of the form \eqref{mini_lagrang}, and used the field space geometry of these Lagrangians to characterize the presence of a CVH algebra and possible extensions thereof. A natural question is therefore that of the relationship between this approach and the standard ADM field theory approach, in which one deals directly with the geometry of the $3+1$ foliation instead of considering the geometry of an auxiliary field space. The spacetime ADM approach enables to show that a CVH algebra always exists if the three-dimensional Ricci scalar of the spatial slice is vanishing, regardless of the value of the cosmological constant. Here we recall this calculation.

Let us denote the spatial indices by $\alpha,\beta,\gamma,\dots$. We consider the standard ADM approach where the canonical variables are the spatial metric $q_{\alpha\beta}$ and its momentum $p^{\alpha\beta}=\sqrt{q}(Kq^{\alpha\beta}-K^{\alpha\beta})$. We denote $q=\det(q_{\alpha\beta})$. The smeared Hamiltonian constraint is
\be\label{ADM hamiltonian}
H=\int_\Sigma\de^3x\,N\cH\,,\q\q\cH=-\sqrt{q}\big(R^{(3)}-2\Lambda\big)-\f{1}{\sqrt{q}}Q_{\alpha\beta\gamma\delta}p^{\alpha\beta}p^{\gamma\delta}\,,
\ee
where $R^{(3)}$ is the Ricci scalar of the three-dimensional slice and
\be
Q_{\alpha\beta\gamma\delta}\coloneqq\f{1}{2}q_{\alpha\beta}q_{\gamma\delta}-q_{\alpha\gamma}q_{\beta\delta}\,.
\ee
From this expression of the ADM Hamiltonian, it is immediate to identify which term gets mapped to the terms of the field space Hamiltonian derived from the Lagrangian \eqref{mini_lagrang}. Under the symmetry reduction to homogeneous minisuperspaces we have indeed that
\be
H=\int_\Sigma\de^3x\,N\cH\q\stackrel{\text{sym. red.}\vphantom{\f{1}{2}}}{\Longrightarrow}\q\cV_0N\left(\f{1}{2}\tilde{g}^{\mu\nu}p_\mu p_\nu+\tilde{U}\right)\,,
\ee
so $R^{(3)}-2\Lambda$ in the ADM Hamiltonian becomes the potential $\tilde{U}$ of the superspace Hamiltonian, while the deWitt supermetric $Q_{\alpha\beta\gamma\delta}$ becomes the field space metric $\tilde{g}_{\mu\nu}$. It should be noted however that this identification, although unambiguous, cannot be made much more explicit. It is really just an identification. This is the reason for which exact calculations on the side of ADM or of the superspace formulation cannot easily be compared to one another.

An example illustrating this is the proof that the CVH algebra always exists when $R^{(3)}=0$. This proof is possible in the ADM language, but not in the superspace formulation. In order to see this, we assume that the vector constraint has been solved by the requirement of homogeneity, although here we keep working in a mixed framework where we keep all the spatial integrals explicit (this is because there might be a spatial dependency in the frames $e^i_\alpha(x)$ introduced in \eqref{minisuperspace}, although these are not the dynamical variables of the homogeneous theory). The Poisson bracket is $\lb q_{\alpha\beta}(x),p^{\gamma\delta}(y)\rb=\delta^\gamma_{(\alpha}\delta^\delta_{\beta)}\delta^d(x,y)$. In order to obtain the CVH algebra and compute Poisson brackets, we define the smeared Hamiltonian and a volume variable by
\be
H=\int_\Sigma\de^3x\,N\cH\,,
\q\q
V=\int_\Sigma\de^3x\, q\,,
\ee
and we choose the lapse to be $N=1/\sqrt{q}$. With this choice of lapse it is then immediate to compute
\be
C\coloneqq\lb V,H\rb=-\f{1}{4}\int_\Sigma\de^3x\,p^{\alpha\beta}q_{\alpha\beta}\,,
\ee
and to show that we have the closed brackets
\be
\lb V,C\rb=-18V\,,\q\q
\lb C,H\rb=-3H+6\Lambda\int_\Sigma\de^3x,
\ee
where in the last bracket the integral (over a finite region with fiducial cut-offs) produces a numerical factor $L_0$ which plays the role of shift in the Hamiltonian when defining $\cQ_0=H-L_0$.

As announced, this simple calculation shows that it is always possible to chose the lapse so as to obtain a closed algebra between $(C,V,\cQ_0)$ when $R^{(3)}=0$. However, this calculation carried out in the ADM formulation cannot be reproduced in the field space formulation since there it is not possible to write an analogue of $q=\det(q_{\alpha\beta})$. Of course, for any minisuperspace model for which $R^{(3)}=0$ the field space calculation will also lead to an $\sl(2,\R)$ algebra formed by $(C,V,\cQ_0)$, and we are guaranteed by the above ADM calculation that this will \textit{always} work, but this has to be computed on a case by case basis.

\subsection{Triad decomposition}
\label{triad_decomp}

The viewpoint we have taken in section \ref{sec:homogeneous} is to define a minisuperspace as a manifold sliced in such a way that the line element separates into a temporal (i.e. orthogonal to the slice) and spatial (i.e. tangential to the slice) dependence as in \eqref{minisuperspace}. This implies that the trace of the extrinsic curvature and the ADM kinetic term depend only on the internal metric $\gamma_{ij}$, up to the determinant of the spatial triad. This latter, once integrated out, gives the volume $\cV_0$ of the fiducial cell. More precisely, we have
\be
K_{\alpha\beta} = \f{1}{2N} \dot q_{\alpha\beta}\q \Rightarrow\q \sqrt{q} \left (K^2 - K^{\alpha\beta}K_{\alpha\beta}\right ) = \f{ |e| \sqrt{\gamma}}{4 N^2} \left ((\gamma^{ij} \dot \gamma_{ij})^2+ \dot \gamma_{ij} \dot \gamma^{ij}\right )\,,
\ee
where $e\coloneqq \det(e_\mu^i)$ and $\gamma\coloneqq\det(\gamma_{ij})$.

In order to analyse the three-dimensional curvature it turns out to be useful to introduce the spin connection
\be
\omega^{ij}_{\alpha}
&\coloneqq e^{\beta i} \partial_{[\alpha}^{\phantom{j}} e_{\beta]}^j -e^{\beta j} \partial_{[\alpha}^{\phantom{i}} e_{\beta]}^i - e^{\delta i} e^{\gamma j} e^k_\alpha \partial_{[\delta} e_{\gamma] k}\cr
&\phantom{:}=\gamma^{\ell i}\big(e^{\beta}_\ell \partial_{[\mu}^{\phantom{j}} e_{\beta]}^j \big)- \gamma^{\ell j} \big(e^{\beta}_\ell \partial_{[\mu}^{\phantom{i}} e_{\beta]}^i\big) -  \gamma_{\ell k}  \gamma^{n i}  \gamma^{m j}\big (e^{\delta}_n e^{\gamma}_m  e^\ell_\mu \partial_{[\delta}^{\phantom k} e_{\gamma]}^k\big)\cr
&\phantom{:}=\gamma^{\ell [i}\big (e^{\beta}_\ell \partial_{\mu}^{\phantom{j}} e_{\beta}^{j]} \big)- \gamma^{\ell [i}\big (e^{\beta}_\ell \partial_{\beta}^{\phantom{j}} e_{\mu}^{j]}\big) -  \gamma_{\ell k}  \gamma^{n [i}  \gamma^{j] m}\big (e^{\delta}_n e^{\gamma}_m  e^\ell_\mu \partial_{\delta}^{\phantom k} e_{\gamma}^k\big)\,.
\ee
One can see that this expression does not simply split into the product of quantities depending separately on the triad and on the internal metric. The same happens for the curvature, which is given by
\be
F^{ij}_{\alpha\beta} \coloneqq 2 \left (\partial_{[\alpha}^{\phantom j} \omega^{ij}_{\beta]} + \gamma_{k \ell} \omega^{i\ell}_{[\alpha} \omega^{kj}_{\beta]} \right ) \,,\q\q
R^{(3)} =\f{1}{|e|} \epsilon^{\alpha\beta\delta}\epsilon_{ijk} e^k_\delta F^{ij}_{\alpha\beta}\,.
\ee
The vector constraint \eqref{vector_constr} depends explicitly on the spin connection and reads
\be
\cH^\alpha
&= 2 D_\beta \left (K q^{\alpha\beta}-K^{\alpha\beta} \right )\cr
&= \f{1}{N} D_\beta \left (\gamma^{k \ell} \dot \gamma_{k \ell} e^\alpha_i\, e^\beta_j\, \gamma^{ij}+ e^\alpha_i\, e^\beta_j\, \dot \gamma^{ij} \right )\cr
&= \f{1}{N} D_\beta \big(e^\alpha_i e^\beta_j\big) \left (\gamma^{k \ell} \dot \gamma_{k \ell} \, \gamma^{ij}+ \dot \gamma^{ij} \right )\cr
&=-\f{\pi^{ij}}{\cV_0 \sqrt{\gamma}} \left (\partial_\beta\big(e^\alpha_i\, e^\beta_j\big) + \big(e^{\alpha}_k e_{\sigma \ell} \omega^{k\ell}_\beta + e^\alpha_k \partial_\beta e^k_\sigma\big) e^\sigma_i\, e^\beta_j + \big(e^{\beta}_k e_{\sigma \ell} \omega^{k\ell}_\beta + e^\beta_k \partial_\beta e^k_\sigma\big) e^\alpha_i\, e^\sigma_j \right ) \cr
&=-\f{\pi^{ij}}{\cV_0 \sqrt{\gamma}} \left (\cancel{\partial_\beta\big(e^\alpha_i\, e^\beta_j\big)} + e^{\alpha}_k \gamma_{i\ell}\, e^\beta_j \omega^{k\ell}_\beta -  \cancel{\partial_\beta e^\alpha_i\, e^\beta_j} +  e^\alpha_i\, \gamma_{\ell j} e^{\beta}_k  \omega^{k\ell}_\beta -  \cancel{\partial_\beta e^\beta_j\, e^\alpha_i} \right )\,,
\ee
where we have introduced the momenta
\be
\pi^{ij}\coloneqq\f{\delta \cL_{\rm {ADM}}}{\delta \dot \gamma_{ij}} = - \cV_0 \f{ \sqrt{\gamma}}{N} \left ((\gamma^{k\ell} \dot \gamma_{k\ell})\,  \gamma^{ij} + \dot \gamma^{ij}\right ) \,,
\ee
which are the conjugate momenta to the internal metric in the ADM form of the action \eqref{Einstein_mini}. The vanishing of the vector constraint is therefore equivalent to the requirement that $\pi_i^j e^{(\alpha}_k e^{\beta)}_j \omega^{ik}_\beta$=0. 

\subsection{Properties of homothetic Killing vectors}
\label{HKV_prop}

In the main text we have used some properties of the conformal Killing vectors, such as the fact that they are solutions to the geodesic deviation equation. We give here a proof of this statement as well as other properties of the conformal Killing vectors.

Given an invertible metric $g_{\mu\nu}$, we recall that the homothetic Killing vectors are defined by the property
\be
\label{HKV_append} \nabla_\mu \xi_\nu + \nabla_\nu \xi_\mu \coloneqq  2  \nabla_{(\mu} \xi_{\nu)} =\lambda g_{\mu\nu}\,,\q\q \lambda=\rm{const}.
\ee
We then have the following result:
\begin{theorem}
Any homothetic killing vector $\xi$ is a solution of the geodesic deviation equation \cite{Caviglia:1982aa}:
\be
\label{geodesic deviation}
p^\mu p^\nu \nabla_\mu \nabla_\nu \xi_\rho = - R_{\rho\mu\sigma\nu}p^\mu p^\nu \xi^\sigma\,,
\ee
where $p^\mu$ is the tangent vector to a geodesic (i.e. a curve describing a solution of the equations of motion). This vector satisfies the property $p^\mu \nabla_\mu p^\nu =0$.
\end{theorem}
\begin{proof}
We start from the definition of the conformal Killing vectors to get
\be
p^\mu p^\nu \nabla_\mu \nabla_\nu \xi_\rho 
&= - p^\mu p^\nu \nabla_\mu \nabla_\rho \xi_\nu + \lambda p^\mu p^\nu \cancel{\nabla_\mu g_{\nu \rho}}\cr
&= - p^\mu p^\nu  R_{\nu\sigma\mu\rho} \xi^\sigma +  p^\mu p^\nu \nabla_\rho \nabla_\mu \xi_\nu\,,
\ee
where the second line is obtained from the definition of the Riemann tensor as a commutator of covariant derivatives. The last term in this expression is now vanishing when $\xi$ is conformal with constant $\lambda$ since it gives
\be
p^\mu p^\nu \nabla_\rho \nabla_\mu \xi_\nu = p^\mu p^\nu \nabla_\rho \nabla_{(\mu} \xi_{\nu)} = \lambda p^\mu p^\nu \nabla_\rho g_{\mu\nu }=0\,.
\ee
This therefore proves the above statement using the following properties of the Riemann tensor:
\be
p^\mu p^\nu \nabla_\mu \nabla_\nu \xi_\rho  = - p^\mu p^\nu  R_{\nu\sigma\mu\rho} \xi^\sigma=
- R_{\rho\mu\sigma\nu}p^\mu p^\nu \xi^\sigma\,.
\ee
\end{proof}
We now give a few properties of the homothetic Killing vectors and of the phase space functions $C_{(i)}$ and $V_{(ij)}$ constructed out of them. First, they form an algebra under the Lie bracket
\be\label{structure}
[\xi_{(i)},\xi_{(j)}] &= {c_{ij}}^k \xi_{(k)}\,,\q\q
[\xi_{(i)},\xi_{(j)}]^\mu\coloneqq \xi_{(i)}^\nu \nabla_\nu \xi_{(j)}^\mu -\xi_{(j)}^\nu \nabla_\nu \xi_{(i)}^\mu\,,
\ee
and the commutator of two vectors is actually a Killing vector, as one can check explicitly by computing
\be
\cL_{[i,j]} g_{\mu\nu} &= \nabla_ \mu (\xi_{(i)}^\sigma \nabla_\sigma \xi_{(j)\,\nu} -\xi_{(j)}^\sigma \nabla_\sigma \xi_{(i)\,\nu}) + \left (\mu \leftrightarrow \nu \right )\notag\\
&= \nabla_ \mu \xi_{(i)}^\sigma \nabla_\sigma \xi_{(j)\,\nu} + \xi_{(i)}^\sigma \nabla_ \mu \nabla_\sigma \xi_{(j)\,\nu} -\nabla_ \mu \xi_{(j)}^\sigma \nabla_\sigma \xi_{(i)\,\nu} + \xi_{(j)}^\sigma \nabla_ \mu \nabla_\sigma \xi_{(i)\,\nu} + \left (\mu \leftrightarrow \nu \right )\notag\\
&= \nabla_ \mu \xi_{(i)\,\nu} \lambda_{(j)} - \cancel{\nabla_ \mu \xi_{(i)}^\sigma \nabla_\nu \xi_{(j)\,\sigma}}  + \xi_{(i)}^\sigma \nabla_ \mu \nabla_\sigma \xi_{(j)\,\nu} -\nabla_ \mu \xi_{(j)\,\nu} \lambda_{(i)}\cr 
&\pe+ \cancel{\nabla_ \mu \xi_{(j)}^\sigma \nabla_\nu \xi_{(i)\,\sigma}} + \xi_{(j)}^\sigma \nabla_ \mu \nabla_\sigma \xi_{(i)\,\nu} + \left (\mu \leftrightarrow \nu \right )\notag\\
&= \cancel{g_{\mu\nu}\lambda_{(i)} \lambda_{(j)}} + \xi_{(i)}^\sigma R_{\nu\rho\mu\sigma}\xi_{(j)}^\rho - \xi_{(i)}^\sigma \nabla_ \sigma \nabla_\mu \xi_{(j)\,\nu} -\cancel{g_{\mu\nu}\lambda_{(j)}\lambda_{(i)}} \cr
&\pe- \xi_{(j)}^\sigma R_{\nu\rho\mu\sigma}\xi_{(i)}^\rho + \xi_{(j)}^\sigma \nabla_ \sigma \nabla_\mu \xi_{(i)\,\nu}+ \left (\mu \leftrightarrow \nu \right )\notag\\
&= \xi_{(i)}^\rho \xi_{(j)}^\sigma (R_{\nu\mu\sigma\rho} + R_{\mu\nu\sigma\rho})  -\xi_{(i)}^\sigma \nabla_\sigma (g_{\mu\nu}) \lambda_{(j)} + \xi_{(j)}^\sigma \nabla_\sigma (g_{\mu\nu}) \lambda_{(i)}\cr
& =0\,.
\ee
Here we have used \eqref{HKV_append} when going from the second to the third line, and eliminated the antisymmetric terms in $\mu\,,\nu$. We have then used \eqref{HKV_append} once again as well as the definition of the Riemannn tensor, and finally concluded by using the antisymmetry of the Riemann tensor. This result implies that the structure constants of the algebra of homothetic Killing vectors satisfy
\be
{c_{ij}}^k \lambda_{(k)} =0\,.
\ee

Let us now consider the phase space functions $V_{(ij)}=g_{\mu\nu}\xi^\mu_{(i)}\xi^\nu_{(j)}$ and $C_{(i)}=p_\mu\xi^\mu_{(i)}$. First, we have that the third time derivative of the squared vectors gives
\be\label{triple dot V Riemann}
\f{\de^3}{\de t^3}V_{(ij)}&=
p^\mu p^\nu p^\rho \nabla_\mu \nabla_\nu \nabla_\rho \left (\xi_{(i)}^\sigma \xi_{(j)\sigma}\right )\notag\\&=2p^\mu p^\nu p^\rho \nabla_\mu\left ( (\nabla_\nu \xi_{(i)}^\sigma) (\nabla_\rho \xi_{(j)\sigma}) -  \xi_{(i)}^\sigma R_{\sigma\nu\kappa\rho} \xi_{(j)}^\kappa) \right ) \notag\\
&= - 2 p^\mu p^\nu p^\rho \left (2(\nabla_\rho \xi_{(j)}^\sigma) R_{\sigma\mu\kappa\nu} \xi_{(i)}^\kappa +2(\nabla_\rho \xi_{(i)}^\sigma) R_{\sigma\mu\kappa\nu} \xi_{(j)}^\kappa  + \xi_{(i)}^\sigma\xi_{(j)}^\kappa (\nabla_\rho R_{\sigma\mu\kappa\nu})\right).
\ee
This shows in particular that $\dddot{V}_{(ij)}=0$ whenever the Riemann tensor vanishes (this is of course sufficient but not necessary), as in the case of the flat field space geometry discussed in section \ref{sec:2d}.

We are interested in the condition for which the functions $V_{(ij)}$ and $C_{(i)}$ form a closed algebra with the Hamiltonian $H$. In the free case where $H=p_\mu p^\mu/2$ we find
\be
\lb V_{(ij)},H\rb
&=p^\mu \partial_\mu \big( \xi_{(i)}^\nu \xi_{(j)\,\nu} \big)\,\notag\\
&=p^\mu \xi_{(i)}^\nu \nabla_\mu   \xi_{(j)\,\nu} +p^\mu \xi_{(j)}^\nu \nabla_\mu   \xi_{(i)\,\nu}\notag\\
&= \lambda_{(i)} C_{(j)}+\lambda_{(j)} C_{(i)}- p^\mu\left ( \xi_{(i)}^\nu \nabla_\nu   \xi_{(j)\,\mu} +\xi_{(j)}^\nu \nabla_\nu   \xi_{(i)\,\mu}\right ) \,,
\ee
and
\be
\lb V_{(ij)}, C_{(k)}\rb
&=\xi_{(k)}^\mu \partial_\mu  \big( \xi_{(i)}^\nu \xi_{(j)\,\nu} \big)\,\notag\\
&=\xi_{(k)}^\mu \xi_{(i)}^\nu \nabla_\mu   \xi_{(j)\,\nu} +\xi_{(k)}^\mu \xi_{(j)}^\nu \nabla_\mu   \xi_{(i)\,\nu}\notag\\
&= \lambda_{(i)} V_{(jk)}+\lambda_{(j)} V_{(ik)}- \xi_{(k)}^\mu\left ( \xi_{(i)}^\nu \nabla_\nu   \xi_{(j)\,\mu} +\xi_{(j)}^\nu \nabla_\nu   \xi_{(i)\,\mu}\right ) \,,
\ee
which closes iff
\be\label{condition on xi's}
\xi_{(i)}^\nu \nabla_\nu   \xi_{(j)\,\mu} +\xi_{(j)}^\nu \nabla_\nu   \xi_{(i)\,\mu}=\sum_k \alpha^k\xi_{(k)\,\mu}\,,
\ee
for some combination of the vectors on the RHS.

\subsection{Bianchi classification}
\label{Bianchi_metrics}

The appendix gathers all the properties of the Bianchi models which are needed for the study of the phase space symmetry algebra. First, we give a list of the Bianchi line elements which satisfy the vector constraint \eqref{vector_constr}. These are\footnote{To this list we should add the model VI$_{h}$, which has the same symmetry properties as VI$_0$ but whose treatment involves some lengthy expressions. It is reported in \cite{francesco-thesis}.}
\bsub
\be
(\rm{I})&&\de s^2 &= -N^2 \de t^2 + a^2\de x^2 + b^2 \de y^2+c^2\de z^2\,,\\
(\rm{II})&&\de s^2 &= -N^2 \de t^2 + a^2 \big(\de x-z\,\de y\big)^2 + b^2 \de y^2+c^2\de z^2\,,\label{bianchi II}\\
(\rm{III})&&\de s^2 &= -N^2 \de t^2 + a^2 \de x^2 + b^2 L_s^2 \big (\de y^2+\sinh^2 y\, \de \phi^2\big )\,,\\
(\rm{IV})&&\de s^2 &= -N^2 \de t^2 + a^2 L_s^2 \de x^2 + \f{a^4}{b^2}e^{-2x}\de y^2 +b^2 e^{-2x}\big(\de z-x\,\de y\big)^2\,,\\
(\rm{V})&&\de s^2 &= -N^2 \de t^2 + a^2 L_s^2 \de x^2 + \f{a^4}{b^2}e^{-2x} \de y^2 +b^2 e^{-2x} \de z^2\,,\\
(\rm{VI}_0)&&\de s^2 &= -N^2 \de t^2 + a^2 L_s^2 \de x^2 + b^2 \big (e^{-2x} \de y^2 +e^{2x}\de z^2\big )\,,\\
(\rm{VII}_h)&&\de s^2 &= -N^2 \de t^2 + a^2 L_s^2 \de x^2  + b^2 e^{-2hx}\big(\cos x\, \de z-\sin x\, \de y\big)^2\cr
&&&\phantom{=\ -N^2 \de t^2 + a^2 L_s^2 \de x^2 }+\f{a^4}{b^2} e^{-2hx}\big(\cos x \, \de y + \sin x \,\de z\big)^2 \,,\\
(\rm{VIII})&&\de s^2 &= -N^2 \de t^2 + a^2\big(\de x + L_s\cosh y\,\de \phi\big)^2  + L_s^2  b^2\, \big(\de y^2 +\sinh y\,  \de \phi\big)^2 \,,\\
(\rm{IX})&&\de s^2 &=-N^2\de t^2+a^2\big(\de x+L_s\cos\theta\,\de\phi\big)^2+L_s^2b^2\, \big(\de \theta^2 +\sin \theta\,  \de \phi\big)^2 \,.
\ee
\esub
The length scale $L_s$ has been introduced in order to have dimensionless fields. In terms of the decomposition \eqref{minisuperspace}, the fundamental triads corresponding to these line elements are
\bsub
\be
&&&e^1=&&e^2=&&e^3=\\
(\rm{I})\quad&&
&\de x &
& \de y &
& \de z \\
(\rm{II})\quad&&
&\de x-z\,\de y &
& \de y &
& \de z \\
(\rm{III})\quad&&
&\de x &
& L_s\, \de y &
& L_s\, \sinh y\,\de \phi \\
(\rm{IV})\quad&&
& L_s\,\de x &
& e^{-x}\, \de y &
&e^{-x}(\de z- x\,\de y)\\
(\rm{V})\quad&&
&L_s\,\de x&
&e^{-x}\, \de y &
&e^{-x}\, \de z \\
(\rm{VI}_0)\quad&&
&L_s\,\de x&
&e^{-x}\, \de y &
&e^{x}\, \de z \\
(\rm{VII}_h)\quad&&
&L_s\,\de x&
&e^{-hx}(\cos x\, \de y+\sin x\, \de z) &
&e^{-hx}(\cos x\, \de z-\sin x\, \de y)\\
(\rm{VIII})\quad&&
&\de x+ L_s \cosh y\, \de z&
&L_s \de y &
&L_s \sinh y\, \de \phi\\
(\rm{IX})\quad&&
&\de x+ L_s \cos y\, \de z&
&L_s  \de \theta &
& L_s  \sin\theta\, \de \phi
\ee
\esub
One should note that for each triad the line elements given above are not the only solutions to the vector constraint. We have focused here on the diagonal case, (i.e. when the internal metric is diagonal), a more involved analysis is needed if we want to account for all the possible internal degrees of freedom \cite{Ashtekar:1991wa}. The finite volumes of the fiducial cells are $\cV_0 = \f{1}{16 \pi}\int_\Sigma |e|$ and given by
\bsub
\be
(\text{I, II})&& \cV_0=\,&\f{1}{16\pi}L_x L_yL_z &&x\in[0,L_x]\,,y\in[0,L_y]\,, z \in[0,L_z]\,, \\
(\text{III, VIII})&& \cV_0=\,&\f{1}{4}L_x L_s^2\sinh^2 \left (\f{y_0}{2}\right) &&x\in[0,L_x]\,,y\in[0,y_0]\,, \phi \in[0,2\pi]\,, \\
(\text{IV, V})&& \cV_0=\,&\f{1}{16 \pi}L_s L_y L_z e^{-x_0}\sinh x_0 &&x\in[0,x_0]\,,y\in[0,L_y]\,, z \in[0,L_z] \,,\\
(\rm{VI}_0)&& \cV_0=\,&\f{1}{16 \pi}L_s L_y L_z x_0  &&x\in[0,x_0]\,,y\in[0,L_y]\,, z \in[0,L_z] \,,\\
(\rm{VII}_h)&& \cV_0=\,&\f{1}{32 \pi h}L_s L_y L_z (1-e^{-2h x_0}) &&x\in[0,x_0]\,,y\in[0,L_y]\,, z \in[0,L_z]\,, \\
(\rm{IX})&& \cV_0=\,&\f{1}{4}L_x L_s^2  &&x\in[0,L_0]\,,\theta\in[0,\pi]\,, \phi \in[0,2\pi]\,. 
\ee
\esub
The models can be divided into three categories depending on the internal metrics, which are
\bsub
\be
&\text{(I, II)} &&& \gamma_{ij} &= \rm{diag}\left(a^2,b^2,c^2\right)\,,\\
&\text{(III, VI$_0$, VIII, IX)} &&& \gamma_{ij} &= \rm{diag}\left(a^2,b^2,b^2\right)\,,\label{second gamma cat}\\
&\text{(IV, V, VII)} && & \gamma_{ij} &= \rm{diag}\left (a^2,\f{a^4}{b^2},b^2\right)\,.
\ee
\esub
We see that for the Bianchi I and II models the field space is three-dimensional, while for the other Bianchi models it is only two-dimensional.

To the second category \eqref{second gamma cat} we can also add the Kantowski--Sachs cosmology, for which the line element, the triad and the fiducial volume are given by
\bsub
\bg
\de s^2_{\rm{KS}} = -N^2 \de t^2 + a^2 \de x^2 + b^2 L_s^2 \left (\de \theta^2+\sin^2 \theta\, \de \phi^2\right)\,,\label{KS metric}\\
e^1 = \de x\,,\q e^2 = L_s\, \de \theta \,,\q e^3 = L_s\, \sin \theta \, \de \phi \,,\\
\cV_0 = \f{1}{4}L_x L_s^2\,, \q \q x\in[0,L_x]\,,\theta\in[0,\pi]\,, \phi \in[0,2\pi].
\eg
\esub
The KS model does not belong to the Bianchi classification because it does not admit three independent spacelike Killing vectors forming a closed Lie algebra. Appart from this difference it fits entirely in the setup of our discussion.

We now give the expression, for each of the above models, of the potential coming from the minisuperspace reduction of the three-dimensional Ricci scalar. This is
\be
\f{1}{16 \pi G}\int_\Sigma\de^3x\, |e| \sqrt{\gamma} \,R^{(3)} =- \f{\cV_0}{G}N\tilde{U}_\text{model}\,,
\ee
with
\be
\tilde{U}_\text{KS}&=-\df{2a}{L_s^2}=-\df{3^{1/3 }}{L_s^2 \sigma^{4/3}}\df{\tilde u}{\tilde{v}^{1/3}}\,,\\
\tilde{U}_\text{I}&=0\,,\\
\tilde{U}_\text{II}&=\df{a^3}{2bc}\,,\\
\tilde{U}_\text{III}&= \df{2a}{L_s^2}=\df{3^{1/3 }}{L_s^2 \sigma^{4/3}}\df{\tilde u}{\tilde{v}^{1/3}}\,,\\
\tilde{U}_\text{IV}&=\df{6a}{L_0^2}+ \df{b^4}{2 L_0^2 a^3 }=\f{1}{4 L_0^2}(3\tilde u \tilde v)^{1/3}\left (12+16^{\sqrt{3}} 81^{1/\sqrt{3}} (\tilde{v} \sigma^2 /\tilde u)^{4/\sqrt{3}}\right )\,,\\
\tilde{U}_\text{V}&=\df{6a}{L_0^2}=\df{3(3\tilde u \tilde v)^{1/3}}{L_0^2}\,,\\
\tilde{U}_\text{VI$_0$}&= \df{2b^2}{L_0^2 a}=\df{(3 \sigma^8 \tilde v^5)^{1/3}}{\tilde u L_0^2}\,,\\
\tilde{U}_\text{VII$_h$}
&= \df{a^8+2(6h^2-1) a^4 b^4 +b^8}{2 L_0^2 a^3b^4}\cr
&=\df{(3\tilde u\tilde v)^{1/3}}{2 L_0^2} (6 h^2 -1) +\df{(3\tilde u \tilde v)^{1/3}}{4 L_0^2}\left (16^{-\sqrt{3}}(3 \tilde{v} \sigma^2 /\tilde u)^{-4/\sqrt{3}}+ 16^{\sqrt{3}}(3 \tilde{v} \sigma^2 /\tilde u)^{4/\sqrt{3}}\right )\,,\q\\
\tilde{U}_\text{VIII}&=\df{a^3+4a b^2 }{2L_s^2 b^2}=\df{3^{1/3}\left (\tilde u^3 + 4 \tilde u \tilde v^2 \sigma^4 \right )}{4 L_s^2 \tilde v^{7/3} \sigma^{16/3}}\,,\\
\tilde{U}_\text{IX}&=\df{a^3-4a b^2}{2L_s^2 b^2}=\df{3^{1/3}\left (\tilde u^3 - 4 \tilde u \tilde v^2 \sigma^4 \right )}{4 L_s^2 \tilde v^{7/3} \sigma^{16/3}}	\,,
\ee
where we have used the map \eqref{u v tilde map} (which is different for the two families of internal metrics) to also express these quantities in the conformal null coordinates.

We see that the potentials of the Bianchi III and KS models differ by the relative sign between potential and the kinetic terms, i.e. of the ``mass of the particle moving on the field space''. The potential in these two cases and in the Bianchi V and VI models have a single term, so they can be removed with a change of lapse while keeping the field space metric flat. This can be achieved with the choice of lapse and coordinates
\be
\hspace{-5cm}\text{(III, KS)}&&N&= \f{\cV_0}{2 G a}\,, &&
\left\{\begin{array}{rlcrl}
u&\!\!\!\!=- 2 a^2 b\,,		& &\hspace{1.4cm} p_u	&\!\!\!\!=-\df{p_a}{4 a b}\,,\\[9pt]
v&\!\!\!\!= 2 b	\,,		& & p_v 				&\!\!\!\!= \df{2 b p_b+a p_a }{4b}\,,\\
\end{array}\right.\\
(\rm{V})&&N&= \f{\cV_0}{6 G a}\,,&&
\left\{\begin{array}{rlcrl}
u&\!\!\!\!= 3 a^{2+\f{2}{\sqrt{3}}}	b^{-\f{2}{\sqrt{3}}}\,,			
					& &\hspace{0.4cm} p_u	&\!\!\!\!=\left(\df{a}{b}\right)^{-2/\sqrt{3}}\df{ap_a+(1-\sqrt{3})bp_b}{12 a^2}\,,\\[9pt]
v&\!\!\!\!=3 a^{2-\f{2}{\sqrt{3}}}b^{\f{2}{\sqrt{3}}}\,,
					& & p_v 				&\!\!\!\!= \left (\df{a}{b}\right)^{2/\sqrt{3}}\df{ap_a+(1+\sqrt{3})bp_b}{12 a^2}\,,\\
\end{array}\right.\\
(\rm{VI})&&N&= \f{\cV_0 a}{2 G b^2} \,,&&
\left\{\begin{array}{rlcrl}
u&\!\!\!\!= 2 \log a^2 b\,,	& &\hspace{1.1cm} p_u 	&\!\!\!\!=\df{a p_a}{2}\,,\\[9pt]
v&\!\!\!\!= b^4	\,,		& &p_v 				&\!\!\!\!=  \df{2 b p_b-a p_a}{8b^4}\,.\\
\end{array}\right.
\ee
We can then use these changes of coordinates in the expression of the generators to obtain the algebra $\cA$ in terms of the scale factors and their momenta.

For the other models it is not possible to obtain a flat field space geometry by choosing the lapse to be the inverse of the potential. However, looking at the potentials for the models VIII and IX, we see that the two monomials are the same and only differ by a sign. Furthermore, they satisfy the condition \eqref{power_condition}. In these two cases we have two possible choices of lapse which lead to the $\sl(2,\R)$ subalgebra of $\cA$. The first choice of lapse and coordinates is
\be
N= \f{\cV_0}{2G a}\,,\q\q u=2 a^2 b\,,\q\q v=2b\,,
\ee
and it leads to the charges
\be
C_\rm{d}=\f{b b_p}{2}\,,\q\q V_3= 8 a^2 b^2 \,,\q\q \cQ_0=\f{p_a (a p_a -2 b p_b)}{16 a b^2}+ \f{\cV_0 a^2}{4  L_s^2 G^2 b^2}\,.
\ee
The second choice of lapse and coordinates is
\be
N= - \f{2\cV_0 b^3}{G a^3}\,,\q\q u=\df{a^4 b^2}{2}\,,\q\q v= \df{1}{2 b^2}\,,
\ee
and it leads to the charges
\be
C_\rm{d}=\f{b b_p}{2}\,,\q\q V_3= \f{a^4}{2}\,,\q\q \cQ_0= -\f{p_a (a p_a -2 b p_b)}{4 a^3}\mp \f{4 \cV_0 b^2}{  L_s^2 G^2 a^2}\,,
\ee
where the $(-)$ sign is for Bianchi VIII and the sign $(+)$ for Bianchi IX. Unfortunately \eqref{power_condition} is not satisfied for the Bianchi models IV and VII, meaning that they do not exhibit the $\sl(2,\R)$ symmetry.

\bibliographystyle{Biblio}
\bibliography{Biblio.bib}

\end{document}